\documentclass[12pt]{article}

\usepackage{amsmath}
\usepackage{amsthm}
\usepackage{amsfonts}
\usepackage[parfill]{parskip}
\usepackage{braket}
\usepackage{mathabx}
\usepackage{dsfont}
\usepackage{graphicx}
\usepackage{float}
\usepackage{color}

\DeclareMathOperator{\Tr}{Tr}
\DeclareMathOperator{\spec}{spec}
\DeclareMathOperator{\supp}{supp}
\DeclareMathOperator{\diam}{diam}
\DeclareMathOperator{\dd}{\mathbf{d}}

\newtheoremstyle{definitionnew}
{15pt} 
{15pt} 
{} 
{} 
{\bfseries} 
{.} 
{.5em} 
{} 

\newtheoremstyle{plainnew}
{15pt} 
{15pt} 
{\itshape} 
{} 
{\bfseries} 
{.} 
{.5em} 
{} 

\newtheoremstyle{remarknew}
{15pt} 
{15pt} 
{} 
{} 
{\bfseries} 
{.} 
{.5em} 
{} 

\theoremstyle{definitionnew}
\newtheorem{definition}{Definition}
\newtheorem{conjecture}{Conjecture}

\theoremstyle{plainnew}
\newtheorem{theorem}{Theorem}
\newtheorem{corollary}{Corollary}
\newtheorem{lemma}{Lemma}
\newtheorem{proposition}{Proposition}

\theoremstyle{remarknew}
\newtheorem{remark}{Remark}

\begin{document}
\centerline{\LARGE {\bf Entanglement Rates and the Stability of the Area Law}}
\centerline{\LARGE {\bf  for the Entanglement Entropy}}
\bigskip

\centerline{\large
Micha\"el Mari\"en$^{1}$,
Koenraad M.R.~Audenaert$^{1,2}$,
Karel Van Acoleyen$^{1}$,
Frank Verstraete$^{1,3}$
}
\medskip

\begin{center}
$^1$\,Department of Physics and Astronomy, Ghent University\\ Krijgslaan 281, S9, 9000 Ghent, Belgium
\end{center}
\begin{center}
$^2$\,Mathematics Department, Royal Holloway University of London\\
Egham TW20 0EX, United Kingdom
\end{center}
\begin{center}
$^3$\,Fakult\"at f\"ur Physik, Universit\"at Wien\\
Boltzmanngasse 5, A-1090 Wien, Austria
\end{center}

\begin{abstract}
\noindent
We prove a conjecture by Bravyi on an upper bound on entanglement
rates of local Hamiltonians. We then use this bound to prove
the stability of the area law for the entanglement entropy of quantum
spin systems under adiabatic and quasi-adiabatic evolutions. 

\end{abstract}

\tableofcontents

\section{Introduction}
\label{sec:intro}
Since its first appearance in  a series of ground breaking papers, 
entanglement has become one of the defining trademarks of quantum mechanics, appearing ubiquitously, both at the theoretical and experimental level.
The main goal of these first efforts to describe and understand this phenomenon was to show the incompleteness of the theory of quantum mechanics. However, entanglement has been experimentally verified and is nowadays considered an important feature of the theory and a valuable resource in quantum information and computation protocols. Because of its importance, it is not surprising that the dynamical properties of entanglement are also of great interest. A very important aspect of a physical system is the rate at which entanglement is, or can be created. 
The first step in many applications, notably in quantum optics, nuclear magnetic resonance and condensed matter physics, is the creation of entanglement, and much experimental effort has been devoted to optimise this process.

The entanglement rate is very important from a theoretical viewpoint as well and can be used in a wide variety of problems, since it is the dynamical version of one of the most fundamental concepts in quantum information theory, the von Neumann entropy. Some applications are straightforward \cite{Vidal,Bennett,cubitt2005entanglement,childs2004reversible,childs2002asymptotic}. Other applications are more surprising. For example in quantum computing one is interested in establishing bounds on the running time and quantum complexity of algorithms. Since these quantum algorithms are based on the phenomenon entanglement, they often generate it themselves as a by-product. Hence bounding how fast this generation can occur can also establish lower bounds on the running times of these algorithms \cite{ambainis2000quantum}.

On the other hand, the entanglement generated in a system can also have deleterious effects. Decoherence \cite{schlosshauer2005decoherence} is a result of the entanglement between  a system and its environment. This effect is very  undesirable in quantum computing, since it destroys the information stored in a coherent superposition of several qubits. It severely shortens the time a quantum computer can use to do reliable calculations and makes it hard to construct robust quantum memories \cite{shor1995scheme}. Bounds on the entanglement generation therefore also yield bounds on the decoherence time.

In the first part of this paper we answer the following fundamental question about entanglement dynamics. Given a Hamiltonian interaction between two subsystems $A$ and $B$, what is the maximal rate at which the corresponding unitary evolution can generate entanglement between these subsystems \cite{Vidal,Bennett,Bravyi}? We provide a tight upper bound on the maximal entanglement rate in the most general setting. 

In the second part of this paper we show that surprisingly, these dynamic properties also have important consequences for the static entanglement properties of quantum many-body systems. Indeed, as an application we prove a stability result for the area law of entanglement entropy. Over the last two decades it has been realized that ground states of gapped local Hamiltonians have very specific properties. One finds that the entanglement entropy of the reduced density matrix of such a ground state of a certain subregion scales with the size of the boundary of this region, instead of the volume law scaling of a typical quantum many body state \cite{page1993average,PhysRevLett.72.1148,hayden2006aspects}. It is exactly this boundary scaling that is referred to as the area law \cite{eisert2010colloquium}. This feature of the entanglement of ground states provides a new window on these systems that allows for a better understanding. Indeed, guided by the area law, efficient tensor network representations of the ground states of quantum many body systems have been proposed. The link between the area law and an efficient representation of the ground state was proven rigorously for one dimensional systems \cite{hastings2007area,verstraete2006matrix,landau2013polynomial}. These insights led to the development of several efficient numerical tensor network methods \cite{verstraete2008matrix}. Moreover, important theoretical advances, for instance the classification of different quantum phases of matter \cite{chen2011complete,schuch2011classifying}, were made through the use of these methods.

However, despite its importance, the area law has only been proven for gapped one-dimensional systems. Our results provide a step in the direction of a general proof for higher dimensions. The formalism of quasi-adiabatic continuation  induces an equivalence relation on the set of ground states of gapped Hamiltonians, these equivalence classes are commonly referred to as gapped quantum phases. Our bound on the entanglement rate allows us to show that for states in the same gapped quantum phase, a subsystem's entropy obeys the same scaling law in both states. Hence, the area law is a property of an entire quantum phase and it suffices to consider one single representative of a phase and show that it satisfies the area law.

Moreover, the stability of the area law of the entanglement entropy validates the use of this entropy as a good measure for quantum many body systems. Indeed, the strongest continuity bound on the entropy is given by the Fannes-Audenaert inequality \cite{fannes1973continuity,audenaert2007sharp}, which has a volume law scaling in the dimensions of the underlying Hilbert space. Unfortunately, we seldom know the exact ground state of a quantum many body system, but instead we have to rely on an approximation, often in the form of a tensor network state. As the Fannes-Audenaert inequality suggests that the entropy of quantum many body states is fragile against perturbations, this inequality is too weak to infer features of the entanglement of the true ground state from the approximation. In contrast, our result implies that such an inference is possible, the scaling of two states in the same phase differs at most by the area instead of the volume, hence giving a quantitative notion of the robustness of the entanglement entropy.

\section{Mixing and Entanglement Rates}\label{sec:twoconjectures}
\subsection{Small Incremental Mixing}\label{sec:SIM}
We start our discussion with a property of ensembles of quantum states known as \textit{small incremental mixing} (SIM). This property was first conjectured by Bravyi in \cite{Bravyi}. Part of the physical relevance of this property lies in its relation to the \textit{small entangling rate property} (SIE), which we introduce in subsection \ref{sec:SIE}. We discuss the connection between SIE and SIM in more detail in subsection \ref{sec:RelatingSIMSIE}.

The situation we consider is the following. Suppose we have a probabilistic ensemble of states $\mathcal{E}$. We immediately restrict ourselves for now to the case with only two different states, $\mathcal{E}= \{(p,\rho_1),(1-p,\rho_2)\}$. The expected state of the system is the convex combination of both, $\rho = p\rho_1+(1-p)\rho_2$. Let each state of this system evolve according to a different Hamiltonian, $H_0$ and $H_1$. Without loss of generality we can assume that $H_1 = 0$ and write $H_0=H$. The evolved expected density operator is then given by
\begin{equation}
\rho(t) = p\rho_1 + (1-p)e^{-iHt}\rho_2e^{iHt}.
\end{equation}
We are interested in the von Neumann entropy of this mixture. It is well known that the entropy of this state remains bounded from below and from above for any time $t$. The precise result is as follows.

\begin{proposition}\label{totalmix}
Let $\mathcal{E} = \{\rho_i,p_i\}$ be a probabilistic ensemble of states and let each state evolve according to a different Hamiltonian $H_i$. Then the entropy of the expected state satisfies 
\begin{equation}
\overline{S} \leq S(\rho(t)) \leq \overline{S} + \mathsf{h}(\{p_i\}).
\end{equation}
Here, $\overline{S} = \sum_i p_iS(\rho_i)$ is the average entropy of the ensemble, which is constant in time, and $\mathsf{h}(\{p_i\})$ is the Shannon entropy of the distribution $\{p_i\}_i$. This property is called \textbf{small total mixing}. 
\end{proposition}

The proof of this proposition relies on well-known properties of the von Neumann entropy. Observe that for 2 states, this proposition implies that the variation of the entropy goes to zero if $p_1$ or $p_2$ goes to zero, as one expects.
\begin{proof}
To prove the lower bound we use the following inequalities,
\begin{equation}
\overline{S} = \sum_ip_iS(\rho_i)
= \sum_i p_iS(\rho_i(t)) 
\leq S\left(\sum_i p_i \rho_i(t)\right) = S(\rho(t))
\end{equation}
where we first used the unitary invariance of the von Neumann entropy and then the concavity.\\
\\
To prove the upper bound we introduce a classical system $C$ and the joint state
\begin{equation}
\rho_{CQ}(t) = \sum_i p_i \ket{i}\bra{i} \otimes \rho_i(t).
\end{equation}
This state is a density operator on the tensor product of a classical system $C$ and the original Hilbert space. The entropy of this state is given by
\begin{equation}
S(\rho_{CQ}(t)) = \mathsf{h}(\{p_i\}) + \overline{S}.
\end{equation}
We now define the following matrix 
\begin{equation}
R(t)=\begin{pmatrix} \sqrt{p_1\rho_1(t)}\\
  \sqrt{p_2\rho_2(t)} \\
  \vdots \end{pmatrix}.
\end{equation}
Both density matrices $\rho_{CQ}(t)$ and $\rho(t)$ are related to the matrix $R(t)$. First, we note that $\rho_{CQ}(t)$ is a pinching of
\begin{equation}
R(t)R(t)^{\dagger} = \sum_{ij}\sqrt{p_ip_i}\ket{i}\bra{j}\otimes\sqrt{\rho_i(t)}\sqrt{\rho_j(t)}.
\end{equation}
From the concavity of the von Neumann entropy it follows that
\begin{equation}\label{prop1rel1}
S(R(t)R^{\dagger}(t)) \leq S(\rho_{CQ}(t)).
\end{equation}
Second, we have that $\rho(t) = R(t)^{\dagger}R(t)$. Moreover, the non zero part of the spectrum of $R(t)R^{\dagger}(t)$ and $R(t)^{\dagger}R(t)$ are equal. This implies that also the von Neumann entropy of both operators is equal,
\begin{equation}\label{prop1rel2}
S(R(t)R^{\dagger}(t)) = S(R(t)^{\dagger}R(t)) = S(\rho(t)).
\end{equation}
Combining the relations \eqref{prop1rel1} and \eqref{prop1rel2}, we obtain the desired upper bound.
\end{proof}
Given Proposition \ref{totalmix} it is natural to consider the immediate change of the entropy instead of the total change. Bravyi \cite{Bravyi} conjectured the following upper bound on this quantity.
\begin{theorem}[SIM]\label{theoremSIM}
Let $\mathcal{E}$ be a probabilistic ensemble of two states, let $\rho(t):=p_1\rho_1 + p_2e^{-iHt}\rho_2e^{iHt}$ be the expected state through time. Then there exists a  constant $c$ such that
\begin{equation}
\Lambda(\mathcal{E},H) := \left.\frac{dS(\rho(t))}{dt}\right|_{t=0} \leq c\|H\|\mathsf{h}(\{p_1,p_2\})
\end{equation}
independent of the ensemble $\mathcal{E}$ and of the details of the Hamiltonian $H$. This property is called \textbf{small incremental mixing (SIM)}.
\end{theorem}
In \cite{Bravyi} this conjecture was proven under certain restrictions. The first full proof for finite Hilbert spaces was given in \cite{Karel}, obtaining $c=9$. A better constant, $c=2$ was proven in \cite{Koenraad}, by a completely different method. Numerical analysis suggests that $c=1$ might be the sharpest constant possible. In this paper we prove Theorem \ref{theoremSIM} in the more general scenario of infinite dimensional separable Hilbert spaces. We employ the same method as in \cite{Karel}, with small adaptations that result in an increase of the constant to $c=11$.

\subsection{SIM for Larger Ensembles}\label{sec:SIMlargerensembles}
Until now we have focused on ensembles consisting of only two states. There are two reasons for this restriction. First of all, our primary interest is not the mixing rate but the entangling rate. The latter typically deals with a bipartition of the system, which corresponds to the mixing rate of just two states. We refer to subsection \ref{sec:RelatingSIMSIE} for a detailed explanation of the relation between SIM and SIE. Second, the general case can easily be deduced from the case with only two states, and the arguments are much clearer presented for a small ensemble. For completeness, we discuss the case of general ensembles in this subsection. We use some results we obtain further on in this paper. This subsection may be skipped on first reading.

In \cite{LiebVer} the following was conjectured with a constant $c=1$ for general probabilistic ensembles that consist of more than two states.
\begin{theorem}
For any probabilistic ensemble $\{(p_i,\rho_i)\}$, an upper bound on the mixing rate is given by a constant times the Shannon entropy of the distribution $\{p_i\}$,
\begin{equation}\label{liebvers}
\Lambda(\mathcal{E}) := \max\left\{|\Lambda(\mathcal{E},H)|: -\mathds{1} \leq H_i \leq \mathds{1}, \forall i\right\} \leq -\mathsf{h}(\{p_i\}) = -c\sum_i p_i\log p_i.
\end{equation}
Here $c$ is a constant independent of the ensemble.
\end{theorem}
Lieb and Vershynina provided an upper bound 
\begin{equation}
\Lambda(\mathcal{E}) \leq 2\sum_i\sum_{j\neq i} \sqrt{p_ip_j}
\end{equation}
that is rather sharp for large values of the $p_i$ but unfortunately fails to capture the behaviour of the conjectured bound \eqref{liebvers} for small values of $p_i$. We give a proof of this bound with a larger constant that nevertheless does capture the small $p_i$ behaviour.
\begin{proof}
We can obtain an explicit expression of the mixing rate by calculating the derivative. We find that
\begin{equation}
\Lambda(\mathcal{E},H) = -i\sum_ip_i\Tr\left(H_i[\rho_i,\log\rho]\right).
\end{equation}
Assuming that SIM for an ensemble of two states holds, which we prove later on, we can apply the upper bound given in equation \eqref{upperboundlog} to each term separately. We restate the bound \eqref{upperboundlog} here for convenience of the reader,
\begin{equation}
\left|\Tr(\log(B)[A,H])\right|\leq 11 p\log(1/p).
\end{equation}
The conditions on the operators $A,B$ can be found in the statement of Theorem \ref{theoremcomm}.
We now have that
\begin{equation}
\Lambda(\mathcal{E},H) \leq 11\sum_i p_i\log\frac{1}{p_i} = 11 \mathsf{h}(\{p_i\}).
\end{equation}
Using the tighter bound obtained by one of the authors \cite{Koenraad}, we obtain a constant 4 instead of 11. The  constant $c=4$ is only valid for finite dimensional Hilbert spaces.
\end{proof}

\subsection{Small Incremental Entangling}\label{sec:SIE}
In this subsection we discuss another property, also first conjectured in \cite{Bravyi}, and attributed by its author to Kitaev. This property is called \textit{small incremental entangling} (SIE) and is physically more relevant than SIM.
We consider two parties, Alice and Bob who both have a Hilbert space, $A$ respectively $B$, at their disposal. 
Clearly, an interaction Hamiltonian $H_{AB}$ can create or destroy entanglement between the two parties. We are interested in the maximal rate at which an interaction can change the bipartite entanglement through time.\\
\\
Without ancillas, a tight upper bound can easily be proven as was done in \cite{Bravyi}, or using different methods in \cite{Hutter}. We consider the more general and interesting case of ancilla assisted entangling. Apart from the systems $A,B$, Alice and Bob each have an ancilla $a,b$ respectively, but the interaction Hamiltonian $H_{AB}$ only acts on the systems $A,B$. The initial state $\ket{\psi_0}$ of the system $aABb$ is assumed to be pure, hence it stays pure throughout its evolution. The state $\ket{\psi_0}$ can have initial entanglement that can dramatically change the rate at which entanglement can be created \cite{Vidal}. Given that we are interested in an upper bound on this rate, independently of the initial state and its entanglement,  we only assume that the initial state is pure. Mathematically, the quantity we study is given by
\begin{equation}
\Gamma(H,\psi) = \left.\frac{dS(\rho_{aA}(t))}{dt}\right|_{t=0}
\end{equation}
where $\rho_{aA}$ is the reduced density matrix of the total state on subsystem $aA$ and
\begin{equation}
\rho_{aA}(t) = \Tr_{Bb} \left(e^{-iHt}\ket{\psi_0}\bra{\psi_0}e^{iHt}\right)  \text{ with } H= \mathds{1}_a \otimes H_{AB} \otimes \mathds{1}_b.
\end{equation}
In the absence of ancillas, Bravyi proved that
\begin{equation}\label{finalbound}
\Gamma(H) \leq c\|H\| \log(d)
\end{equation}
with $c$ a constant and $d$ the smallest of the dimensions of systems $A,B$. Moreover, he conjectured that this bound also holds in the presence of ancillas but could not prove it.\\
\\
Since this problem has significant importance in the optimal creation of entanglement, it has been studied by several authors. In \cite{Vidal} the optimal initial state to create entanglement in the absence of ancillas was identified. Moreover, the authors found that in general, ancillas can increase the entanglement rate, showing the relevance of the presence of ancillas. In \cite{childs2002asymptotic}, the asymptotic entanglement rate of interaction Hamiltonians between two qubit systems was studied in detail. In \cite{wang2002entanglement}, arbitrary dimensions and ancillas were considered with the interaction given by a self-inverse product Hamiltonian. Under these restrictions, the bound $\Gamma(H) \leq \beta$, with $\beta \approx 1.9123$, was obtained. This result was generalized to arbitrary bipartite product Hamiltonians in \cite{childs2004reversible}. The first general bound independent of the dimension of the ancillas was proven in \cite[bound 4]{Bennett}. For a general Hamiltonian $H_{AB}$ the authors argued that $\Gamma(H)\leq cd^4\|H\|$ with $d = \min(A,B)$ and $c$ a constant, independent of the ancillas $a,b$. The work of Bravyi \cite{Bravyi} implies a refinement, of the form $\Gamma(H)\leq 2\|H\|d^2$. Finally, the results obtained in \cite{LiebVer} imply a bound of the form $\Gamma(H) \leq 4/\ln(2)\|H\|d$, which is still exponentially weaker than the conjectured bound \eqref{finalbound} for large values of $d$. In this paper we prove the logarithmic bound \eqref{finalbound}, see Theorem \ref{SIE}, which is known to be tight.\\
\\
To motivate the conjectured bound \eqref{finalbound} we first prove a simple upper bound on the total change of the entropy.
As in the case of SIM, we can bound the total change of entanglement throughout time.
\begin{proposition}\label{totalent}
For the system described above, the total change of entanglement through time is bounded from above by 
\begin{equation}\label{eqtotalent}
\Delta S(\rho_{aA}(t)) \leq 2\log d 
\end{equation}
with $d = \min\{\dim A, \dim B\}$. This property is called \textbf{small total entangling}.
\end{proposition} 
\begin{proof}
The proof is based on the following observation \cite{Bennett}. Every non-local unitary gate $U_{AB}$ can be simulated by first teleporting system $A$ to system $B$, perform the gate and teleport $A$ back. The amount of entanglement consumed in such a double teleportation is exactly $2\log A$. 

Alternatively, we can give a proof based on the following inequalities. Suppose $d=\dim(B)\leq \dim(A)$. Denote by $\rho_{aA}$ and $\tilde{\rho}_{aA}$ the reduced density matrix of the system before and after applying the unitary $U_{AB}$ respectively. We have that
\begin{align}
\left|S(\rho_{aA}) - S(\tilde{\rho}_{aA})\right| &=
\left|S(\rho_{aA})-S(\rho_{aAB})+S(\tilde{\rho}_{aAB}) - S(\tilde{\rho}_{aA})\right|\\
&\leq \left|S(\rho_{aA})-S(\rho_{aAB})\right| + \left|S(\tilde{\rho}_{aAB}) - S(\tilde{\rho}_{aA})\right|\\
& \leq S(\rho_B) + S(\tilde{\rho}_B)\\
&\leq 2\log(d).
\end{align}
In the second line we used the fact that $S(\rho_{aAB})=S(\tilde{\rho}_{aAB})$ since $U_{AB}$ does not act on system $b$. The last inequality follows from the subadditivity of the von Neumann entropy. The bound \eqref{eqtotalent} is tight, as setting $U_{AB}$ equal to the swap gate shows. 
\end{proof}
Kitaev proposed a related upper bound on the maximal rate at which the entanglement can change.
\begin{theorem}[SIE]\label{SIE}
Denote $d = \min\{\dim A, \dim B\}$, then there is a constant $c$ such that
\begin{equation}\label{SIEbound}
\Gamma(H,\psi) \leq c\|H\|\log d
\end{equation}
independently of the dimensions of the ancillas $a,b$, the initial state $\ket{\psi}$ of $aABb$ and the details of the interaction Hamiltonian $H$. We call this property \textbf{small incremental entangling (SIE)}.
\end{theorem}
As it is implied that $H$ is an interaction Hamiltonian only acting on $A,B$, we dropped the explicit subscript for clarity. Bravyi proved this upper bound under certain restrictions. Using his method, a proof for this conjecture was first obtained in \cite{Karel}. As we explain in subsection \ref{sec:RelatingSIMSIE}, our method gives a constant $c$ that is twice the constant obtained for the SIM Theorem \ref{theoremSIM}. In this paper we obtain a constant $c=22$. Based on numerical evidence, we expect $c=2$ to be the optimal value.

Bravyi already showed that the logarithmic scaling in the bound \eqref{SIEbound} is tight; surprisingly this is already true for systems without ancillas \cite{Bravyi}.

One can wonder about the importance of using ancillas. As already noted, examples are known where the use of ancillas allows for a larger entangling rate. Furthermore, several phenomena have been studied in the literature where local ancillas and some initial entanglement between a part of the system and an ancilla can have unexpected results. The well known example of the swap operator illustrates this for the entanglement rate \cite{Bennett}. If no ancillas are present this operator cannot change the entanglement between $A$ and $B$. However, if both $A,B$ have an identical copy as ancilla and we start from the state $\ket{M_{aA}}\otimes\ket{M_{Bb}}$ with $\ket{M}$ a maximally entangled state, it is clear that the swap operator creates the maximal amount of entanglement as given in Proposition \ref{totalent}. 

There are other manifestations of the importance of ancillas. For example in \cite{locking} it was shown that the mutual information can violate a property called incremental proportionality. In the presence of ancillas it is possible to increase the classical mutual information between two parties by an arbitrary amount just by sending a single qubit. In contrast, the mutual information does satisfy a property called total proportionality, which can be considered the analogue of Proposition \ref{totalmix} and \ref{totalent}. At first sight, there is no reason to believe the same cannot happen for the von Neumann entropy. Is it possible to lock entanglement in a state, such that in a short time an interaction can free this entanglement and hence allow for an arbitrary change in the entropy? As we shall prove, this is not the case.

\subsection{Relating SIM and SIE}\label{sec:RelatingSIMSIE}
At first sight, SIM and SIE look like two rather unrelated dynamical properties of the entanglement entropy. In this subsection we show that SIM is actually a stronger version of SIE. This connection was made by Bravyi in \cite{Bravyi}, but for completeness we repeat his argument here in detail.\\
\\
To make the connection between the quantities $\Lambda$ and $\Gamma$, we first give explicit expressions for both. An easy calculation shows that the mixing rate is given by
\begin{equation}\label{mixingrate}
\Lambda(\mathcal{E},H) = -ip_2 \Tr\left(H[\rho_2,\log(\rho)]\right) \text{ with } \rho := p_1\rho_1+p_2\rho_2.
\end{equation}
Similarly we can work out the derivative and obtain an expression for the entangling rate,
\begin{equation}\label{entrate}
\Gamma(H,\psi) = -i\Tr\left(\mathds{1}_a \otimes H_{AB} [\rho_{AaB},\log(\rho_{Aa})\otimes \mathds{1}_B]\right).
\end{equation}
Without loss of generality we assume that $d = \dim B \leq \dim A$. Since the bound in Theorem \ref{SIE} only depends on the smallest dimension, we can extend $A$ to $a \otimes A$, hence we can assume that $\dim a = 1$ and reduce the total system to $ABb$. We now define the state 
\begin{equation}
\tau_{AB} = \rho_A \otimes \frac{\mathds{1}_B}{d}
\end{equation}
and rewrite
\begin{equation}\label{SIErewritten}
\Gamma(H,\psi) = -i\Tr\left(H[\rho_{AB},\log\tau_{AB}]\right).
\end{equation}
This last expression already starts to look like the mixing rate \eqref{mixingrate}. We continue by defining a well suited probabilistic ensemble $\mathcal{E}$ to complete the reduction of SIE to SIM. By comparing the equations \eqref{mixingrate} and \eqref{SIErewritten} it clear that we want $\tau_{AB}$ to be the expected state of an ensemble of which $\rho_{AB}$ is one of the constituents. The following simple lemma, taken from \cite{Bravyi}, shows that such an ensemble exists.
\begin{lemma}[\cite{Bravyi}]\label{extens}
Let $\rho_{AB}$ be a mixed state. Then there exists another mixed state $\mu_{AB}$ such that
\begin{equation}
\rho_A \otimes \frac{\mathds{1}_B}{d} = \frac{1}{d^2} \rho_{AB} + \left(1-\frac{1}{d^2}\right)\mu_{AB}.
\end{equation}
\end{lemma}
\begin{proof}
Clearly the existence of the state $\mu_{AB}$ is equivalent to the condition
$$
\rho_{AB} \leq  d(\rho_A \otimes \mathds{1}_B).
$$
Furthermore, since the partial trace is a linear operation, if suffices to consider the case that $\rho_{AB} = \ket{\psi}\bra{\psi}$, a pure state. Define the maximally entangled unnormalized state
\begin{equation}
\ket{I} = \sum_{j=1}^d\ket{j}_A\ket{j}_B
\end{equation}
with $\{\ket{j}_A\}$ a basis of the Hilbert space $A$ and $\{\ket{j_B}\}$ of the Hilbert space $B$. By the Schmidt decomposition there exist an operator $X$ and a unitary $U$ such that 
\begin{equation}
\ket{\psi}=X\otimes U\ket{I}.
\end{equation}
Since $\braket{I,I}=d$, we have that $\ket{I}\bra{I} \leq d\mathds{1}_{AB}$. Conjugating this inequality with $X\otimes U$ we immediately get 
\begin{equation}
\ket{\psi}\bra{\psi} \leq d(X^{\dagger}X \otimes \mathds{1}_B) = d(\rho_A \otimes \mathds{1}_B)
\end{equation}
which finishes the proof.
\end{proof}
We can now continue our strategy to show that SIM implies SIE. We define the ensemble 
\begin{equation}
\mathcal{E} = \left\{\left(\frac{1}{d^2},\mu_{AB}\right),\left(1-\frac{1}{d^2}, \rho_{AB}\right)\right\}
\end{equation}
whose existence is assured by Lemma \ref{extens}. Indeed, $\mu_{AB}$ is precisely the state appearing in that lemma. Moreover, the average state of this ensemble is exactly $\rho_A \otimes \mathds{1}_B/d$.
\begin{proposition}
The small incremental mixing theorem implies the small incremental entangling theorem.
\end{proposition}
\begin{proof}
Let us assume that SIM is true. We are clearly only interested in the case where $p\leq 1/2$. To sharpen the constant we assume the conjecture in the following slightly adapted form, 
\begin{equation}\label{SIMonlylog}
\Lambda(\{p,1-p\},H) \leq -cp\log(p)\|H\| .
\end{equation}
This inequality is clearly equivalent with the SIM conjecture, apart from a possible modification of the constant. We use inequality \eqref{SIMonlylog} because it allows for the smallest constant prefactor.

Inequality \eqref{SIMonlylog} gives for the ensemble $\mathcal{E}$ that
\begin{equation}\label{Lambdashannon}
\Lambda(\mathcal{E},H) \leq \frac{2}{d^2}\log(d)\|H\|
\end{equation}
with $H:= H_{AB}$ the Hamiltonian from \eqref{SIErewritten}. Using \eqref{mixingrate} we can write the expression for $\Lambda$ as
\begin{equation}\label{Lambdaensemble}
\Lambda(\mathcal{E},H) =  -i\frac{1}{d^2}\Tr\left(H[\rho_{AB},\log\tau_{AB}]\right).
\end{equation}
We combine equation \eqref{Lambdashannon} and \eqref{Lambdaensemble} and find that
\begin{equation}
-i\Tr\left(H[\rho_{AB},\log\tau_{AB}]\right) \leq 2c\log(d)\|H\|.
\end{equation}
Comparing the expression on the left hand side with \eqref{SIErewritten} we find that
\begin{equation}
\Gamma(H,\psi) \leq 2c\log(d)\|H\|.
\end{equation}
Hence, we can conclude that SIM with a constant $c$ implies SIE with a constant $2c$. 
\end{proof}
It remains to show the validity of the SIM conjecture, in the form of the inequality \eqref{mixingrate} or  \eqref{SIMonlylog}.

\section{A Trace Norm Inequality for Commutators}\label{sec:prooftraceinequality}
In this section we discuss Theorem \ref{theoremcomm}, which was first conjectured in \cite{Bravyi}. From expression \eqref{mixingrate} it is clear that the statement of this theorem is equivalent with the small incremental mixing property, Theorem \ref{theoremSIM}. For the sake of clarity, we state it as an independent result that may be of interest in matrix analysis. Therefore, we prove this result in the more general case of separable Hilbert spaces, although for the physically relevant spin systems, it suffices to deal with finite dimensional Hilbert spaces.
\begin{theorem}\label{theoremcomm}
Let $A,B$ be positive trace class operators on a separable Hilbert space $\mathcal{H}$, such that $A \leq B$, $\Tr(A) = p \in [0,1]$, $\Tr(B) = 1$. Then there is a constant $c$ such that
\begin{equation}\label{matin}
\left\|[A,\log B]\right\|_1 \leq c\: \mathsf{h}(p).
\end{equation}
\end{theorem}
As noted in \cite{AudKitt}, the more general case without extra restrictions on $\Tr(A),\Tr(B)$ can easily be reduced to the inequality \eqref{matin}. Recently a proof yielding $c=9$ was obtained in \cite{Karel} for $\mathcal{H}$ finite dimensional and a completely different proof, based on the continuity properties of the quantum skew divergence, was given in \cite{Koenraad}. The latter proof gives a constant $c=2$. The proof given here is a generalisation of the method in \cite{Karel} and gives a constant $c = 11$. The increase of the constant from 9 to 11 is entirely due to the use of infinite dimensional Hilbert spaces. We remark that numerical evidence suggests that $c=1$ is in fact the best possible constant.\\
\\
To be complete we mention the following conjecture, which is a natural generalisation of Theorem \ref{theoremcomm}.
\begin{conjecture}[\cite{AudKitt}]
Let $A$ and $B$ be positive semidefinite $d\times d$ matrices with $\Tr A = a$ and $\Tr B = b$. For certain functions $f: \mathbb{R} \rightarrow \mathbb{R}$ (still to be determined), there exists a constant $c$, independent of $d$ such that
\begin{equation}
\|[B,f(A+B)]\|_1 \leq c\left(F(a+b)-F(a)-F(b)\right),
\end{equation}
with $F(x) = \int_0^xf(y)dy$.
\end{conjecture}
Theorem \ref{theoremcomm} concludes the proof of Theorem \ref{theoremSIM} (SIM). It also finishes the argument given in section \ref{sec:RelatingSIMSIE}, hence the proof of Theorem \ref{SIE} (SIE). 

We now give the proof of Theorem \ref{theoremcomm} in the next two subsections. Both can be skipped entirely by readers only interested in the result and the applications.
For clarity, we start with some minor lemmas that allow us to deal with a very simplified version of the problem. The idea of the full proof is to reduce the problem to this simplified case. 
\subsection{Some Technical Lemmas}\label{sec:technicallemmas}
We start with some lemmas that can be used to treat the case $\Tr(A) \leq \lambda_{min}(B)$. This makes the extra constraint $A \leq B$ redundant.

\begin{lemma}\label{lemmaweak}
Suppose $A$ is a positive trace class operator and $B$ is a positive operator such that $\spec{B} \subset [b_L,b_U]$. We have that
\begin{equation}\label{matrixinequality}
\|[A,\log B]\|_1\leq 2\Tr A \log\left(\frac{b_U}{b_L}\right).
\end{equation}
\end{lemma}
\begin{proof}
We have that
\begin{align}
\|[A,\log B]\|_1 &= \|[A,\log B-\log b_L\mathds{1}]\|_1\\
&\leq 2\|A(\log B - \log b_L\mathds{1})\|_1\\
&\leq 2 \|A\|_1\|\log B - \log b_L\mathds{1} \|_{op}\\
&= 2\Tr(A)(\log b_U - \log b_L).
\end{align}
In the first line we use the invariance of a commutator under adding a scalar operator to one of the arguments. The second and third line follow from the triangle and H\"older's inequality. In the last line we use the positivity of $A$ and the restriction on $\spec(B)$.
\end{proof}
Since we want to obtain a proof for a constant $c$ as small as possible, we give a stronger statement that removes the factor 2. The proof is analogous but uses a commutator inequality for positive operators by Kittaneh \cite{Kittpos} instead of the triangle inequality.
\begin{lemma}\label{lemmastrong}
Suppose $A$ is a positive trace class operator and $B$ is a positive operator such that $\spec{B} \subset [b_L,b_U]$. We have that
$$
\|[A,\log B]\|_1\leq  \log\left(\frac{b_U}{b_L}\right)\Tr A.
$$
\end{lemma}
Our general strategy is to reduce the general case to the case where we can use Lemma \ref{lemmastrong} and bound the remaining cases with a Cauchy-Schwarz inequality. We use similar ideas as in \cite{Karel}. The main difference with the finite dimensional case is that we need the following theorem by Kittaneh \cite{Kittstrong}. 
\begin{lemma}[\cite{Kittstrong}]\label{CSKit}
Let $A,B$ be bounded self-adjoint operators such that  $\spec(A) \subset [a_L,a_U]$ and $\spec(B) \subset [b_L,b_U]$. Then, for every operator $X$ and every unitarily invariant norm $\vvvert.\vvvert$ we have that
\begin{equation}
\vvvert AX-XB \vvvert \leq \max(a_U-b_L,b_U-a_L)\vvvert X\vvvert.
\end{equation}
\end{lemma}
For completeness we reproduce the short proof.
\begin{proof}
We introduce the notation 
\begin{equation}
a_M = \frac{a_L+a_U}{2} \text{ and } b_M = \frac{b_L+b_U}{2}.
\end{equation}
Then we have that
\begin{align}
\vvvert AX-XB\vvvert &= \vvvert(A-a_M)X-X(B-b_M)+(a_M-b_M)X\vvvert\\
&\leq \left(\|A-a_M\|+\|B-b_M\|+|a_M-b_M|\right)\vvvert X\vvvert\\
&=\left((a_U-a_M)+(b_U-b_M)+|a_M-b_M|\right)\vvvert X\vvvert\\
&=\max(a_U-b_L,b_U-a_L)\vvvert X\vvvert.
\end{align}
\end{proof}
We need a last technical lemma that is a matrix version of the Cauchy-Schwarz inequality. 
\begin{lemma}\label{CSmat}
Let $\bigoplus_k S_k, \bigoplus_k T_k$ be block diagonal Hilbert Schmidt operators. Then, the following holds,
\begin{equation}
\sum_k \left|\Tr(S_k T_k) \right| \leq \left(\sum_k \left\|S_k\right\|_2^2\right)^{1/2}\left(\sum_k \left\|T_k\right\|_2^2\right)^{1/2}.
\end{equation}
\end{lemma}
\begin{proof}
The proof is immediate,
\begin{align}
\sum_k \left|\Tr(S_k T_k) \right| &= \left|\Tr\bigoplus_k S_k \bigoplus_k T_k\right|\\
&\leq \left\|\bigoplus_k S_k\right\|_2 \left\|\bigoplus_k T_k\right\|_2 \\
&= \left(\sum_k \left\|S_k\right\|_2^2\right)^{1/2}\left(\sum_k \left\|T_k\right\|_2^2\right)^{1/2}.
\end{align}
\end{proof}
We use Lemma \ref{CSmat} to replace the Cauchy-Schwarz argument used on the matrix elements in the finite dimensional case \cite{Karel}.
\subsection{Proof of Theorem \ref{theoremcomm}}\label{sec:prooftheorem}
After the previous technical intermission, we return to the main subject of this section, the proof of the matrix inequality \eqref{matin}.
\begin{proof}
To prove Theorem \ref{theoremcomm} we first fix $\Tr(A) = p \in (0,1)$ and partition the spectrum of $B$ in countably many subsets related to this value $p$. To be specific, consider the intervals $I_k = \left[p^{k+1},p^k\right)$  for all $k \in \mathbb{N}$. Notice that we can always restrict the Hilbert space to the support of $B$. Furthermore, since $B$ is positive and has trace equal to 1, 1 itself cannot be in the spectrum of $B$. Hence the union of these intervals $I_k$ ultimately contains the entire spectrum of $B$. Of course, some $I_k$ may be empty; let $K \subset \mathbb{N}$ be the set of integers $k$ for which $\spec(B) \cap I_k \neq \emptyset$.\\
\\
We now use the orthonormal basis consisting of eigenvectors of $B$ and the spectral partitioning $\{I_k\}_{k\in K}$ to decompose the Hilbert space $\mathcal{H}$. Let $\mathcal{H}_k$ be the subspace of $\mathcal{H}$ spanned by the eigenvectors of $B$ that correspond to eigenvalues in $I_k$. This induces a direct sum decomposition
\begin{equation}
\mathcal{H} = \bigoplus_{k\in K} \mathcal{H}_{k}.
\end{equation}
By definition of the direct summands, the operator $B$ also decomposes  as a block diagonal operator
\begin{equation}
B = \bigoplus_{k\in K} B_k
\end{equation} 
where each of the operators $B_k$ only acts on $\mathcal{H}_{k}$. We now introduce the resolution of the identity related to this decomposition. Let $\{P_k\}_{k\in K}$ be the complete set of mutually orthogonal projectors such that $P_k: \mathcal{H} \rightarrow \mathcal{H}_{k}$ is the projector onto $\mathcal{H}_{k}$ and $\sum_{k \in K}P_k = \mathds{1}$ with $\mathds{1}$ the identity operator on the full Hilbert space $\mathcal{H}$. By definition of these projectors we have that $P_kBP_l = B_k\delta_{kl}$. The reason for this decomposition is that the spectrum of the restricted operators $B_k$ is bounded from below and above as $\spec\left(B_k\right) \subset I_k$.

We now use the following variational characterization of the trace norm for a self-adjoint trace class operator $O$, 
\begin{equation}
\|O\|_1 = \max_{-\mathds{1} \leq H \leq \mathds{1}} \Tr(HO).
\end{equation}
We have that
\begin{align}
\|[A,\log B]\|_1 &= \|i[A,\log B]\|_1\\
&= \max_{-\mathds{1} \leq H \leq \mathds{1}} i\Tr([A,\log B]H)\\
&= \max_{-\mathds{1} \leq H \leq \mathds{1}} i\Tr(\log B[A,H]).
\end{align}
Take a random $H$, it suffices to prove that
\begin{equation}
W := i\Tr\left(\log(B)[A,H]\right) \leq c \mathsf{h}(p).
\end{equation}
We now write $W$ as a sum of several terms, where each term has contributions based on the partitioning of $\mathcal{H}$ introduced above. Let us define
\begin{equation}
A_{kl} := P_kAP_l, \quad H_{lk} := P_lHP_k
\end{equation}
and 
\begin{equation}
W_{kl} := i\Tr\left(\log(B_k)A_{kl}H_{lk}-H_{lk}A_{kl}\log(B_l)\right).
\end{equation}
This notation allows us to write
\begin{equation}
\begin{aligned}
W &= i\Tr\left(\log B\left(\sum_k P_k\right)A\left(\sum_l P_l\right)H-\log B\left(\sum_l P_l\right)H\left(\sum_k P_k\right)A\right)\\
&=\sum_{k,l \in K} W_{kl}.
\end{aligned}
\end{equation}
We continue the strategy of expressing everything in the basis of $B$. Since $0\leq A \leq B$ there exists an $0 \leq X \leq \mathds{1}$ such that 
\begin{equation}
A = B^{1/2}XB^{1/2}. 
\end{equation}
This implies that  $A_{kl} = B_k^{1/2}X_{kl}B_l^{1/2}$  with $X_{kl} = P_kXP_l$.\\
\\
We now introduce the central idea to bound $W$. We make a distinction between couples of parts of the spectrum $k,l$ which are close together, i.e. in the same or neighbouring intervals $I_k,I_l$, and those which are far from each other. More specifically, we split the sum as
\begin{equation}\label{rearrangement}
W = \underbrace{\sum_{k,l \in K, |k-l|<2} W_{kl}}_{W'} + \underbrace{\sum_{k,l\in K, |k-l|\geq 2} W_{kl}}_{W''}.
\end{equation} 
The first sum contains the contributions of pairs of eigenvalues close to each other, while the second contains those of further separated pairs. The logic behind the particular rearrangement \eqref{rearrangement} is that the terms in the first sum are the terms that can be bounded using Lemma \ref{lemmastrong}. The terms in the second sum, $W''$, are precisely those that can be bounded using the Cauchy-Schwarz inequality. We proceed by bounding the latter terms, $W''$. The general case considered here requires a bit more care than the finite dimensional case \cite{Karel}. This causes an increase of the constant $c$ from $9$ to $11$.

We introduce some extra operators to lighten the notation. Let 
\begin{equation}
Z_{kl} = B^{1/2}_kX_{kl}, \quad Y_{kl} = X_{kl}B_l^{1/2} = B^{-1/2}_kZ_{kl}B_l^{1/2}.
\end{equation}
We only consider the indices $(k,l)$ such that $l > k+1$, call this set $I$. We now have that
\begin{equation}
W'' := \sum_{(k,l) \in I} (W_{kl} + W_{lk}).
\end{equation} 
From Lemma \ref{CSmat} it follows that
\begin{align}
W'' &\leq 2 \sum_{(k,l) \in I}|W_{kl}|\\
&= 2 \sum_{(k,l) \in I} \left|\Tr\left(\log(B_k)B_k^{1/2}X_{kl}B_l^{1/2}H_{lk}-H_{lk}B_k^{1/2}X_{kl}B_l^{1/2}\log(B_l)\right)\right|\\
&=2 \sum_{(k,l) \in I} \left|\Tr\left[\left(\log(B_k)Y_{kl}-Y_{kl}\log(B_l)\right)\left(H_{lk}B_k^{1/2}\right)\right]\right|\\
&\leq 2\left(\sum_{(k,l) \in I}\left\|\log(B_k)Y_{kl}-Y_{kl}\log(B_l)\right\|_2^2\right)^{1/2}\left(\sum_{(k,l)\in I}\left\|H_{lk}B_k^{1/2}\right\|_2^2\right)^{1/2}.
\end{align}
Recall that by definition, $p^{k+1} \leq B_k \leq p^k$ for all $k \in K$. We find the following inequality
\begin{equation}
\|Y_{kl}\|_2 \leq \|B_k^{-1/2}\|\|Z_{kl}\|_2\|B_l^{1/2}\| \leq \sqrt{p^l/p^{k+1}}\|Z_{kl}\|_2.
\end{equation}
By Lemma \ref{CSKit} we have that,
\begin{align}
\|\log(B_k)Y_{kl}-Y_{kl}\log(B_l)\|_2 &\leq \max\left(\log(p^k/p^{l+1}), \log(p^l/p^{k+1})\right)\|Y_{kl}\|_2\\
&\leq \max\left(\log(1/p^{l-k+1}),\log(p^{l-k-1})\right)\sqrt{p^{l-k-1}}\|Z_{kl}\|_2.
\end{align}
Since we only sum over indices $(k,l) \in I$ for which $l > k+1$, we have that $p^l<p^{k+1}, 1\leq 1/p^{l-k+1}$ and $p^{l-k-1} \leq 1$. Thus, we can bound the prefactor in the previous inequality as
\begin{align}
\max\left(|\log(1/p^{l-k+1})|,|\log(p^{l-k-1}|\right))\sqrt{p^{l-k-1}} &= \log(1/p^{l-k+1})\sqrt{p^{l-k-1}} \\
&= \frac{l-k+1}{l-k-1}\log(1/p^{l-k-1})\sqrt{p^{l-k-1}}\\
&\leq 3\log(1/p^{l-k-1})\sqrt{p^{l-k-1}}.
\end{align}
Due to the slightly different Cauchy-Schwarz argument, here a factor 3 appears, which differs from the finite dimensional case and results in a bigger constant $c=11$.\\
\\
To bound the contributions of the form $\log(1/x)\sqrt{x}$ we consider the function $$x\mapsto \log(1/x)\sqrt(x),$$ which is monotonously increasing on the interval $[0,e^{-2}]$ and attains its maximum value, $2e^{-1}$ at $x_{\text{max}} = e^{-2}$. Since $l-k-1 \geq 1$, we have that
\begin{equation}
\log(1/p^{l-k-1})\sqrt{p^{l-k-1}} \leq f(p)
\end{equation}
with the function $f$ defined as
\begin{equation}
f(p) := 
\begin{cases}
   \log(1/p)\sqrt(p) & \text{if } 0<p\leq e^{-2} \\
  2e^{-1}      &  \text{otherwise}.
  \end{cases}
\end{equation}
For applications, we are mainly interested in the regime $p \ll 1$, since $1/p$ corresponds to $\dim(B)^2$, the dimension of the Hilbert space. Moreover for large $p$, better bounds can be established \cite{LiebVer} that do not suffer from a relatively large constant prefactor. Clearly, it is the first case in the definition of $f$ that is important.\\
\\
We now bound the contribution of $W''$. By the previous observations, we have that
\begin{equation}
W'' \leq 6f(p)\left(\sum_{(k,l)\in I}\|Z_{kl}\|^2_2\right)^{1/2}\left(\sum_{(k,l)\in I}\|H_{lk}B_k^{1/2}\|_2^2\right)^{1/2}.
\end{equation}
Now the initial condition $0\leq A \leq B$ gives that $0\leq X\leq \mathds{1}$, which immediately implies that $0 \leq X^2 \leq X$. Therefore, we have that
$$
\sum_{l \in L}X_{kl}(X_{kl})^{\dagger} \leq X_{kk}
$$
for any possible index set $L \subset K$. We now have that
\begin{align}
\sum_{(k,l)\in I} \|Z_{kl}\|_2^2 &= \sum_{(k,l) \in I} \Tr\left(X_{kl}^{\dagger}B_kX_{kl}\right)\\
&\leq \sum_{k\in K} \Tr(B_k X_{kk})\\
&= \sum_{k\in K}\Tr A_{kk} = \Tr A = p.
\end{align}
We continue in the same fashion to bound the final factor. As we considered normalised interactions, we have that $\|H\| = 1$ and $0 \leq H^2 \leq \mathds{1}$. Therefore, 
\begin{equation}
\sum_{l\in L} H_{lk}^{\dagger}H_{lk} \leq \mathds{1}
\end{equation}
for every set $L \subset K$. Hence we find that
\begin{align}
\sum_{(k,l)\in I}\|H_{lk}B_k^{1/2}\|_2^2 &= \sum_{(k,l)\in I} \Tr\left(H_{lk}^{\dagger}H_{lk}B_k\right)\\
&\leq \sum_{k\in K}\Tr B_k = \Tr B = 1.
\end{align}
Combining these estimates, we find that
\begin{equation}\label{bound1}
|W''| \leq 6f(p)\sqrt{p}.
\end{equation}
We now bound the first part of the sum  \eqref{rearrangement}, $W'$. These are actually the easy terms that can be treated as the case in Lemmas \ref{lemmaweak},\ref{lemmastrong}.
We first need to split up the first term $W'$ even more. Define the set
\begin{equation}
K' = \{k \left|\right. k,k+1 \in K\}
\end{equation}
and
\begin{equation}
V_k := 
\begin{cases}
   W_{k,k}+W_{k,k+1}+W_{k+1,k}+W_{k+1,k+1} & \text{if } k \in K' \\
   W_{k,k}       & \text{if } k \in K \setminus K'.
  \end{cases}
\end{equation}
We can now rewrite the first term as 
\begin{equation}
W' = V - V' = \sum_{k\in K}V_k - \sum_{k \in K'} W_{k+1,k+1}.
\end{equation}
Here we introduce an extra term $V' = \sum_{k \in K'} W_{k+1,k+1}$ to compensate the double counting of some of the diagonal elements $W_{k,k}$. We have obtained the finale decomposition of $W = V-V'+W''$. By the triangle inequality,
\begin{equation}\label{bound0}
W \leq |V| + |V'| + |W''|
\end{equation}
and since we already obtained a bound on $W''$, it suffices to bound the first two terms separately.\\
\\
We first deal with the term $|V|$. Once again, we first introduce some notation. Let us define the projector 
\begin{equation}
Q_k := 
\begin{cases}
   P_k\oplus P_{k+1} & \text{if } k \in K' \\
   P_k      & \text{if } k \in K \setminus K'.
  \end{cases}
\end{equation}
Now we define
\begin{equation}
\tilde{B}_k = Q_kBQ_k, \quad \tilde{A}_{kl} = Q_kAQ_l, \quad \tilde{H}_{kl} = Q_kHQ_l.
\end{equation}
Since $p^{k+2} \leq \tilde{B}_k \leq p^k$, we still have good bounds on the spectrum of $\tilde{B}_{kk}$, although slightly weaker than in the case of the operator $B_{kk}$. We now write the contribution of $V_k$ as
\begin{equation}
V_k = i\Tr\left(\log\tilde{B}_k\tilde{A}_{k,k}\tilde{H}_{k,k}-\tilde{H}_{k,k}\tilde{A}_{k,k}\log\tilde{B}_k\right)
\end{equation}
independently of $k \in K'$ or not, which was the motivation behind the introduction of the projectors $Q_k$. Since $\|H\| \leq 1$, we have for all $k$ that $\|\tilde{H}_{k,k}\| \leq 1$. Hence,
\begin{equation}
V_k \leq \|[\log\tilde{B}_k,\tilde{A}_{k,k}]\|_1.
\end{equation} 
We now apply Lemma \ref{lemmastrong} and $p^{k+2} \leq \tilde{B}_k \leq p^k$ to conclude that
\begin{equation}
V_k \leq \Tr \left(\tilde{A}_{k,k}\right)\log\left(\frac{1}{p^2}\right).
\end{equation}
To bound the contribution of $|V|$ we sum over all $k$ and since all terms are positive, we find that
\begin{align}
V &\leq \log\left(\frac{1}{p^2}\right)\sum_{k\in K} \Tr\tilde{A}_{k,k}\\
&\leq  \log\left(\frac{1}{p^2}\right) 2 \sum_{k\in K} \Tr A_{k,k}\\
&= 4p\log\left(\frac{1}{p}\right).\label{bound2}
\end{align}
The extra factor 2 in the second line appears because of the double counting of some of the diagonal elements, i.e. when $k \in K'$.\\
\\
With a similar reasoning we can bound the contribution of $|V'|$. We have that
\begin{equation}
W_{k,k} \leq \|[\log B_k,A_{k,k}]\|_1 \leq \Tr\left(A_{k,k}\right) \log\left(\frac{1}{p}\right).
\end{equation}
We can now sum over $k \in K'$ and obtain
\begin{align}
V' &\leq \log\left(\frac{1}{p}\right)\sum_{k\in K'} \Tr A_{k+1,k+1}\\
&\leq \log\left(\frac{1}{p}\right)\sum_{k\in K} \Tr A_{k,k}\\
& = \log\left(\frac{1}{p}\right)\Tr A\\
&= p\log\left(\frac{1}{p}\right).\label{bound3}
\end{align}
We obtained the final upper bound on $|W|$. Indeed, putting all obtained upper bounds \eqref{bound0}, \eqref{bound1}, \eqref{bound2}, \eqref{bound3} together, we find that
\begin{equation}\label{boundmetf}
|W| \leq 6\sqrt{p}f(p) + 5p\log(1/p).
\end{equation}
As claimed, for $p \leq e^{-2}$, which is the regime of interest, we find that
\begin{equation}\label{upperboundlog}
|W| \leq 11p\log(1/p).
\end{equation}
More generally, for $p\leq 1/2$ one can easily show that the bound \eqref{boundmetf} is itself smaller than $11\mathsf{h}(p)$ which proves Theorem \ref{theoremcomm} for $p \leq 1/2$ and with $c=11$.\\
\\
We can transform the case $1/2\leq p < 1$ to the discussed case $0<p\leq 1/2$ by  using the substitution $A \mapsto B-A$ and using the fact that $\Tr(B-A) = 1-p$ and $[A,\log(B)] = -[B-A,\log(B)]$. Hence, the claim holds for all $p \in [0,1]$.
\end{proof}

\section{The Stability of the Area Law in a Gapped Phase}\label{sec:Stability}
In this section we deal with quantum spin systems. Let us introduce these systems in more detail. A quantum spin system is defined on an underlying set of vertices $\mathcal{L}$, commonly referred to as sites. The set of vertices $\mathcal{L}$ is referred to as lattice. For simplicity we restrict ourselves in this paper mainly to finite subsets $\mathcal{L} = \mathbb{Z}_L^{\nu}$ of the $\nu$-dimensional integer lattice $ \mathbb{Z}^{\nu}$, $\nu \in \mathbb{N}$. The sites $v$ can be denoted by their coordinates $v=(v_1,\ldots, v_{\nu})$. As the notation $\mathbb{Z}_L^{\nu}$ suggests, we assume periodic boundary conditions, hence $\mathcal{L}$ has the structure of a $\nu$-dimensional torus. Given such a subset, we obtain a quantum spin system by attaching to every vertex $v \in \mathcal{L}$ a $d$-dimensional Hilbert space $\mathcal{H}_v \cong \mathbb{C}^d$. The restriction to isomorphic Hilbert spaces can be removed, although a uniform upper bound on the local dimension is required for some of the arguments. The Hilbert space $\mathcal{H}_{\mathcal{L}}$ of the lattice $\mathcal{L}$ is defined as 
\begin{equation}
\mathcal{H}_{\mathcal{L}} = \bigotimes_{v\in \mathcal{L}}\mathcal{H}_v.
\end{equation}

We need a metric on the set $\mathcal{L}$. There are a few natural and equivalent metrics one can equip $\mathcal{L} \subset \mathbb{Z}^{\nu}$ with, like the Manhattan metric or the shortest path metric. In this paper, we use the shortest path metric. We denote the shortest path distance between two points $x,y \in \mathcal{L}$ as $\dd(x,y)$. Other metrics are denoted by $d(x,y)$.

\begin{definition}\label{metric}
Let $x,y \in \mathcal{L}$ and denote the coordinates of $x$ by $(x_1,\ldots, x_{\nu})$ and of $y$ by $(y_1,\ldots, y_{\nu})$. Then, the \textbf{shortest path distance} $\dd(x,y)$ is defined as
\begin{equation}
\dd(x,y)=\sum_{i=1}^{\nu} \min_{n\in\mathbb{N}} \left|x_i-y_i+nL\right|.
\end{equation}
The distance between two subsets $X,Y \subset \mathcal{L}$ is defined as $\dd(X,Y) =\min_{x\in X, y\in Y} \dd(x,y)$. The \textbf{diameter} of a subset $X\subset \mathcal{L}$ is defined as $\text{diam}(X)=\max_{x_1,x_2 \in X}\dd(x_1,x_2)$. The \textbf{ball} centred at $v_0\in \mathcal{L}$ with radius $r$ is defined as $B_r(v_0) = \{w\in \mathcal{L} \:|\: \dd(v_0,w) \leq r\}$.
\end{definition}
Given a metric, we can define the boundary of a given subset $\mathcal{V} \subset \mathcal{L}$ and the area of the boundary between two subsets.
\begin{definition}\label{defarea}
The \textbf{neighbourhood} $N(v)$ of a site $v \in \mathcal{L}$ is defined as 
\begin{equation}
N(v) = \{w\in \mathcal{L} \:|\: \dd(v,w) = 1\}.
\end{equation}
The \textbf{boundary} $\partial\mathcal{V}$ of a subset $\mathcal{V} \subset \mathcal{L}$ is defined as the points of $\mathcal{V}$ that have a neighbour not in $\mathcal{V}$,
\begin{equation}
\partial\mathcal{V} = \{v\in \mathcal{V} \:|\: N(v) \cap (\mathcal{L} \setminus \mathcal{V}) \neq \emptyset \}.
\end{equation}
Let $\mathcal{B}_1,\mathcal{B}_2$ be a bipartition of $\mathcal{L}$, $\mathcal{L} = \mathcal{B}_1 \bigsqcup \mathcal{B}_2$. The size $A$ of the \textbf{area} of the boundary of this bipartition is defined as $A=\max(|\partial \mathcal{B}_1|,|\partial \mathcal{B}_2|)$. 
\end{definition}

\begin{remark}\label{remark1}
We can consider more general lattices and different metrics.
The most important property we need for the lattice and the metric is that the volume of a ball with radius $r$ does not increase too fast as $r$ increases. More precisely, we require the existence of a polynomial $P(r)$ such that
\begin{equation}\label{polyball}
\max_{v\in\mathcal{L}}|B_r(v)| \leq P(r).
\end{equation}
Clearly, $\mathcal{L} = \mathbb{Z}_L^{\nu}$ equipped with the metric of Definition \ref{metric} satisfies \eqref{polyball} with $P(r) = (2r)^{\nu}$.
\end{remark}

Although we only consider finite sets, quantum spin systems can be rigorously defined and used in the thermodynamic limit. For infinite systems non-trivial conditions on the lattice and the metric are needed to use a similar formalism as in finite Hilbert spaces; see \cite{bratteli1981operator} for more details. These conditions are all satisfied for $\mathbb{Z}^{\nu}$ equipped with the metric $d$. To work in the thermodynamic limit, one considers the relevant algebra of observables. For finite lattices the approach based on Hilbert spaces and on the algebra of observables are equivalent. The algebra of observables associated to a given site $v \in \mathcal{L}$ is given by $\mathcal{A}_v := B(\mathcal{H}_v) \cong M_d(\mathcal{C})$. The algebra of observables of the entire lattice $\mathcal{L}$ is given by
\begin{equation}
\mathcal{A}_{\mathcal{L}} = \bigotimes_{v\in\mathcal{L}} \mathcal{A}_v.
\end{equation}
The \textbf{support} $\text{supp}(A)$ of an operator $A \in \mathcal{A}_{\mathcal{L}}$ is defined as the smallest set of sites on which $A$ acts non trivially. If $\text{supp}(A) = \mathcal{V} \subset \mathcal{L}$, then $A \in \mathcal{A}_{\mathcal{V}}:=\bigotimes_{v\in\mathcal{V}}\mathcal{A}_v$.

It is often useful to define an potential $\Phi$ that generates the Hamiltonian $H$ of a quantum spin system. This is especially convenient if we want to consider the same type of interaction on lattices $\mathcal{L}$ of different sizes, or to rigorously study quantum spin systems in the thermodynamic limit \cite{bratteli1981operator}. Given a lattice, a potential is a map $\Phi$ from the finite subsets $\mathcal{V}$ of this lattice to the operator algebra, $
\Phi: \mathcal{V} \mapsto \Phi(\mathcal{V}) \in \mathcal{A}_{\mathcal{V}}$
such that $\Phi(\mathcal{V})$ is Hermitian for all finite  $\mathcal{V}$. The Hamiltonian of $H$ on $\mathcal{L}$ is then defined as
\begin{equation}
H=\sum_{\mathcal{V} \subset \mathcal{L}}\Phi(\mathcal{V}).
\end{equation}
If we define a potential $\Phi$ on $\mathbb{Z}^{\nu}$, we can use it to generate Hamiltonians $H_{\mathcal{L}}$ for all lattices $\mathcal{L} \subset \mathbb{Z}^{\nu}$.

Since $\mathcal{L}$ has periodic boundary conditions, there exist well defined shift operators 
$T_k$ that map $\mathcal{A}_{(v_1,\ldots,v_k,\ldots,v_{\nu})}$ to $\mathcal{A}_{(v_1,\ldots,v_k+1,\ldots,v_{\nu})}$ for every direction $k \in \{1,\ldots,\nu\}$. We use these elementary shift operators to define the operator $T_{\vec{e}}$ for every direction $\vec{e} \in \mathbb{Z}^{\nu}$.

We now turn our attention to the interactions on quantum spin systems. We need several restrictions on the type of interactions we consider. Most importantly, we require interactions to be local. For several applications a gap between the lowest eigenvalues and the rest of the spectrum is also required. We now define these notions rigorously.
\begin{definition} 
Suppose we have a quantum spin system defined on a lattice $\mathcal{L}$. A \textbf{strictly local}, bounded Hamiltonian $H$ with range $R$ is a Hamiltonian that can be written as a sum of terms $h_v$ with $v \in \mathcal{L}$, where each term $h_v$ only acts non-trivially on sites $w \in B_{R}(v)$ for a fixed, finite $R\geq 0$. Moreover, we require that the norm of the local terms $h_v$ is uniformly bounded by a constant $C$:
\begin{equation}
H= \sum_{v \in \mathcal{L}} h_v, \quad \text{supp}(h_v) \subset B_R(v), \quad \|h_v\| \leq C.
\end{equation}
A Hamiltonian is \textbf{quasi-local} with decay function $f$ if it can be written as
\begin{equation}
H=\sum_{v\in \mathcal{L}}\sum_{r \in \mathbb{N}} h_v(r), \quad \text{supp}(h_v(r)) \subset B_r(v), \quad \|h_v(r)\| \leq f(r).
\end{equation}
If we do not specify the decay function $f$, we assume that $f$ decreases super-polynomially in $r$.
We call a unitary local or quasi-local if it is generated by a local or quasi-local Hamiltonian respectively. A local or quasi-local potential is one that generates local or quasi-local Hamiltonians.
A potential $\Phi$ is called \textbf{translation invariant} if $T_{\vec{e}}\Phi(\mathcal{N})T_{\vec{e}}^{\dagger} = \Phi(\mathcal{V}+\vec{e})$ for all subsets $\mathcal{V}\subset \mathcal{L}$ and all directions $\vec{e} \in \mathbb{Z}^{\nu}$. Here we use the notation $\mathcal{V}+\vec{e}=\{v \in \mathcal{L}\:|\: v - \vec{e} \in \mathcal{V}\}$. 
\end{definition}

Given a translation invariant potential $\Phi$ on the lattice $\mathbb{Z}^{\nu}$, we can use it to obtain Hamiltonians $H_L$ for all lattices $\mathcal{L} = \mathbb{Z}_L^{\nu}$ for all values of $L$.  
These Hamiltonians are itself translation invariant and can be decomposed as $H_L=\sum_{v \in \mathcal{L}} h_v$ with $T_{\vec{e}}h_vT_{\vec{e}}^{\dagger}=h_{v+\vec{e}}$ for all sites $v\in \mathcal{L}$ and all directions $\vec{e} \in \mathbb{Z}^{\nu}$.

Next, we define the notion of a gapped Hamiltonian. Thereto, we need to consider the same interaction on lattices $\mathcal{L}$ of increasing size. Hence, it is natural to consider translation invariant Hamiltonians. Indeed, translation invariant Hamiltonians can naturally be defined on lattices $\mathcal{L} = \mathbb{Z}_L^{\nu}$ for all sizes $L$. Indeed, we can define a translation invariant potential $\Phi$ on the infinite lattice $ \mathbb{Z}^{\nu}$.  We can now look at the behaviour of the sequence of  Hamiltonians $H_L:= \sum_{\mathcal{V} \subset \mathcal{L}}\Phi(\mathcal{V})$ defined on lattices $\mathcal{L} = \mathbb{Z}_L^{\nu}$ of increasing size $L$.

\begin{definition}
Let $\Phi$ be a potential on $\mathbb{Z}^{\nu}$ and denote by $H_L$  the translation invariant Hamiltonians generated by $\Phi$ defined on the Hilbert space associated with the lattice $\mathbb{Z}_L^{\nu}$ for all values of $L$. Then we call the Hamiltonians $H_L$ and the interaction $\Phi$  \textbf{gapped} with ground state degeneracy $q$ if the following two conditions are satisfied. First, the ground state of $H_L$ is $q$-fold degenerate if there is a constant $q\in \mathbb{N}$ such that the  $q$ lowest eigenvalues $E_{0,1}(L),\ldots, E_{0,q}(L)$ of $H_L$  satisfy
\begin{equation}
\delta E= \max_{k,k'}{|E_{0,k}(L)-E_{0,k'}(L)|} \rightarrow 0 \text{ as } L \rightarrow \infty.
\end{equation}
Second, the distance between the ground state sector ${E_{0,1},\ldots, E_{0,q}}$ and the rest of the spectrum is larger than a positive constant $\Delta$ which is independent of $L$. The constant $\Delta$ is called the \textbf{spectral gap}.
\end{definition}

Given the lattice $\mathcal{L}$ equipped with the metric $\dd$ and a strictly local Hamiltonian $H$, one can prove that the time evolution of a strictly local observable $A$ under the evolution generated by $H$ is still approximately local after a finite time $t$. We need this property in subsection \ref{sec:exactQAC}. The property is reminiscent of the concept of strict light cones in relativistic theories. The precise statement is given below in Theorem \ref{theoremliebrobinson}. We discuss it in a general setting since this allows us to apply our results to other quantum spin systems than $\mathcal{L}$ equipped with the metric $\dd$.

For local or quasi-local interactions with exponential decay function on graphs $\mathcal{L} \subset \mathbb{Z}^{\nu}$ the following important theorem holds \cite{liebrobinson}.
\begin{theorem}[Lieb-Robinson]\label{theoremliebrobinson}
Let $\mathcal{L}$ be a lattice equipped with a metric $d$ and a potential $\Phi$. Suppose that for all sites $v \in \mathcal{L}$, the following holds:
\begin{equation}
\sum_{\mathcal{V} \ni v} \|\Phi(\mathcal{V})\||\mathcal{V}|\exp(\mu\diam(\mathcal{V})) \leq s<\infty.
\end{equation}
for some positive constant $\mu,s$. Take a finite subset $\mathcal{W} \subset \mathcal{L}$ and let $H$ be the Hamiltonian generated by $\Phi$ on $\mathcal{W}$. Let $A_X,B_Y$ be local operators supported on disjoint finite sets $X,Y \subset \mathcal{W}$, respectively. Denote the time evolution of $A$ by $\tau^{H}_t(A) := e^{-iHt}Ae^{iHt}$. Then,
\begin{equation}
\|[\tau^H_t(A_X),B_Y]\| \leq 2\|A_X\|\|B_Y\||X|\exp(2s|t|-\mu d(X,Y)).
\end{equation}
\end{theorem} 
This theorem quantifies the speed at which information can propagate through the system \cite{bravyi2006lieb}. The effective speed is given by $2s/\mu$. Clearly, the quantum spin system defined by $\mathcal{L} = \mathbb{Z}_L^{\nu}$ equipped with the metric $\dd$ and a strictly local Hamiltonian satisfies the conditions in Theorem \ref{theoremliebrobinson}. 
Theorem \ref{theoremliebrobinson} can be used to prove the existence of dynamics associated with a potential $\Phi$ in the thermodynamic limit \cite{bratteli1981operator,nachtergaele2006propagation,nachtergaele2010existence,nachtergaele2009lieb}.

It turns out that Theorem \ref{theoremliebrobinson} also holds for more general quasi-local interactions $H$ on more general lattices \cite{liebrobinson,hastings2006spectral,nachtergaele2006lieb}.  
We state the extended theorem since we need it to generalise our results to a broader family of quantum spin systems.
\begin{theorem}\label{theoremLB2}
Let $\mathcal{L}$ be a lattice equipped with a metric $d$ and $\Phi$ a potential. Suppose there exists  a positive real function $K$ such that for all $v,w \in \mathcal{L}$ we have
\begin{equation}\label{reproducing}
\sum_{x\in\mathcal{L}}K(d(v,x))K(d(x,w)) \leq \lambda K(d(v,w))
\end{equation}
for some constant $\lambda$. Furthermore, suppose that for all $v,w \in \mathcal{L}$,
\begin{equation}
\sum_{\mathcal{V}\ni v,w}\|\Phi(\mathcal{V})\| \leq K(d(v,w)) 
\end{equation}
for $\mathcal{V}$ a finite subset of $\mathcal{L}$. Take a finite subset $\mathcal{W} \subset \mathcal{L}$ and let $H$ be the Hamiltonian generated by $\Phi$ on $\mathcal{W}$. Let $A_X,B_Y$ be local operators supported on disjoint finite sets $X,Y \subset \mathcal{W}$, respectively. Then,
\begin{equation}\label{LRgen}
\|[\tau^H_t(A_X),B_Y]\| \leq 2\|A_X\|\|B_Y\||X||Y|K(d(X,Y))\frac{\exp(2\lambda|t|)}{\lambda}.
\end{equation}
\end{theorem}
Functions $K$ that satisfy inequality \eqref{reproducing} are called reproducing and where introduced in \cite{hastings2010quasi}. For $\mathcal{L} = \mathbb{Z}_L^{\nu}$ equipped with the shortest distance metric $d$, $K(r)=r^{-a}$ is reproducing for sufficiently large $a$. The exponential function $K(r) = e^{-r}$ is not reproducing, but $K(r)=e^{-r}r^{-a}$ is reproducing for $a$ large enough. Hence, an exponential with a smaller decay is reproducing. In the literature \cite{hastings2009quantization,hastings2010locality,nachtergaele2006lieb,nachtergaele2006propagation}, small adaptations of Theorem \ref{theoremLB2} appear. Most notably, with slightly different conditions on the functions $K$, the upper bound can depend on $\min(|X|,|Y|)$ instead of on the product $|X||Y|$. These differences are not important for the remainder of this paper.

\begin{definition}
Suppose we have a quantum spin system on a lattice $\mathcal{L}$ equipped with a metric $d$ and let $\Phi$ be a potential on $\mathcal{L}$. We call such a system \textbf{LR-local} with decay function $K$ if $K$ satisfies the conditions in Theorem \ref{theoremLB2}, which implies that inequality \eqref{LRgen} holds for all Hamiltonians $H_{\mathcal{W}}$ generated by $\Phi$ on finite subsets $\mathcal{W} \subset \mathcal{L}$. We only consider the case where $K$ decreases at least super-polynomially in $r$ and we refer to such systems simply as LR-local.
\end{definition}

\begin{remark}
There is a close connection between the notions of quasi-locality and LR-locality. Indeed, suppose we have a lattice $\mathcal{L}$ equipped with a metric $d$ such that the volume of balls with radius $r$ only increases polynomially in $r$. This requirement was discussed in Remark \ref{remark1}. Suppose that $\Phi$ is a potential such that the system is LR-local with a super-polynomial decay function $K$. Let $H = \sum_v \sum_r h_v(r)$ be the Hamiltonian generated by $\Phi$ on some finite subset $\mathcal{W} \subset \mathcal{L}$. Then we have for a fixed $v\in\mathcal{W}$ that
\begin{equation}
\left\|\sum_{r>R}h_v(r)\right\| \leq \sum_{w:d(w,v)>R}\sum_{\mathcal{V} \ni w,v}\|\Phi(\mathcal{V})\| 
\leq \sum_{r>R} |B_v(r)|K(r).
\end{equation}
Here, the summations are restricted to sites $w \in \mathcal{W}$ and $\mathcal{V} \subset \mathcal{W}$. Since $K$ decays super-polynomially and $|B_v(r)|$ only increases polynomially in $r$, this last summation still decays super-polynomially in $r$ with decay function $\tilde{K}$. We see that the LR-locality of this system implies that the Hamiltonian itself is quasi-local. In many applications the decay functions $K,\tilde{K}$ will be very similar \cite{hastings2010quasi}. Without loss of generality we will assume that they are equal and use the notation $K$. This can be achieved by using a decay function that dominates both $K,\tilde{K}$. Clearly, strictly local Hamiltonians are always LR-local.
\end{remark}

\subsection{Entanglement Generation in Quantum Spin Systems}\label{sec:entgeneration}
We first give a very simple application of the SIE bound \eqref{SIEbound} as an introduction to the main application and as a demonstration of the importance of the logarithmic dependence of this upper bound.
We consider the lattice $\mathcal{L}=\mathbb{Z}^{\nu}$  equipped with the metric $\dd$ and a quasi-local, translation invariant Hamiltonian $H$. For this application we only need the property discussed in Remark \ref{remark1} and the quasi-local properties of $H$, the LR-locality of this system is not needed. The restriction to translation invariant interactions can be removed, especially when considering more general lattices. .

We  first define the following quantities. Let $m(v) = \dd(v,\partial \mathcal{B}_1)$ for $v \in \mathcal{B}_2$ and $m(v) =\dd (v,\partial \mathcal{B}_2)$ for $v \in \mathcal{B}_1$ and let $M(r) = \{v\in \mathcal{L} \:|\: m(v) \leq r \}$. This last set contains  exactly the sites of the lattice whose distance to the boundary between $\mathcal{B}_1$ and $\mathcal{B}_2$ is at most $r$. It clear that $M(r) \leq \sum_{v\in \partial \mathcal{B}_1}|B_r(v)|+\sum_{v\in \partial \mathcal{B}_2}|B_r(v)|$. Hence for $\mathbb{Z}^{\nu}$ equipped with the metric $\dd$, it holds that $M(r) \leq 2A(2r)^{\nu}$. More generally, we can prove the following proposition if $M(r) \leq cAr^{\mu}$ for constants $c>0$ and $\mu \geq 0$. The existence of such constants is clear for all lattices that satisfy the condition in Remark \ref{remark1}.

\begin{proposition}\label{entraterealtime}
Consider the lattice $\mathbb{Z}^{\nu}$ equipped with the metric $\dd$ and let each site support a Hilbert space of dimension $d$. Let $\Phi$ be a translation invariant potential that generates a quasi-local Hamiltonian $H_L$ on the lattice  $\mathcal{L}=\mathbb{Z}_L^{\nu}$ for all $L$ with decay function $f$ that decreases faster than $r^{-(2\nu+1+\delta)}$ for a $\delta >0$.
Take $\mathcal{L}=\mathbb{Z}_L^{\nu}$ for a fixed $L$ and consider a bipartition $\mathcal{B}_1,\mathcal{B}_2$ of the system.  Denote the size of the area of the boundary between $\mathcal{B}_1$ and $\mathcal{B}_2$ by $A$. Then the entanglement rate of $H$ relative to this bipartition satisfies an area law, 
\begin{equation}
\left|\frac{dS_{\mathcal{B}_1}(s)}{ds}\right| \leq CA
\end{equation}
with $C$ a constant that depends on the details of the lattice, the metric and the Hamiltonian, but, importantly, not on the lattice size $L$.
\end{proposition}
\begin{proof}
By definition $H$ can be decomposed as
\begin{equation}
H=\sum_{v \in \mathcal{L}} \sum_{r \in \mathbb{N}} h_v(r) \text{ with } \|h_v(r)\| \leq f(r).
\end{equation}
  Let us assume the system is in a pure state $\ket{\psi}$ and denote $\rho_{\mathcal{B}_1} = \Tr_{\mathcal{B}_2}\ket{\psi}\bra{\psi}$. By definition, the rate at which the Hamiltonian $H$ creates entanglement between $\mathcal{B}_1$ and $\mathcal{B}_2$ is given by
\begin{equation}\label{sumS}
\frac{dS_{\mathcal{B}_1}(s)}{ds} = i \sum_{v\in\mathcal{L}} \sum_{r\in\mathbb{N}}\Tr\left(h_v(r)[\ket{\psi}\bra{\psi},\log \rho_{\mathcal{B}_1} \otimes \mathds{1}_{\mathcal{B}_2}]\right).
\end{equation} 
It is clear that operators $h_v(r)$ that only act within $\mathcal{B}_1$ or $\mathcal{B}_2$ do not contribute to this rate. Indeed, suppose $h$ only acts within $\mathcal{B}_2$, then $[\rho_{\mathcal{B}_1},h]=0$. It follows immediately that 
\begin{equation}
\Tr\left(h[\ket{\psi}\bra{\psi},\log \rho_{\mathcal{B}_1} \otimes \mathds{1}_{\mathcal{B}_2}]\right) = 0
\end{equation}
using the cyclicity of the trace. If $h$ only acts on $\mathcal{B}_1$, similar arguments allow us to conclude the contribution of $h$ to the sum \eqref{sumS} vanishes. Indeed, if $h$ is only supported on $\mathcal{B}_1$ we have that
\begin{equation}
\Tr\left(h[\ket{\psi}\bra{\psi},\log \rho_{\mathcal{B}_1} \otimes \mathds{1}_{\mathcal{B}_2}]\right) = 
\Tr_R\left(h[\rho_{\mathcal{B}_1},\log\rho_{\mathcal{B}_1}]\right)=0.
\end{equation}

Therefore we can restrict the summation over $v$. We are interested in interaction terms $h_{v}(r)$ whose range $r$ is larger than the distance of the site $v$ to the boundary of the bipartition. We have that
\begin{align}\label{arearate}
\left|\frac{dS_{\mathcal{B}_1}(s)}{ds}\right| &\leq\sum_{v\in\mathcal{L}} \sum_{r\geq m(v)}\left|\Tr\left(h_{v}(r)[\ket{\psi}\bra{\psi},\log \rho_{\mathcal{B}_1} \otimes \mathds{1}_{\mathcal{B}_2}]\right)\right|.\\
&\leq
\sum_{r\in \mathbb{N}} \sum_{v\in \mathcal{L} : m(v) \leq r}\left|\Tr\left(h_{v}(r)[\ket{\psi}\bra{\psi},\log \rho_{\mathcal{B}_1} \otimes \mathds{1}_{\mathcal{B}_2}]\right)\right|\\
& \leq  \sum_{r\in\mathbb{N}} |M(r)| \left( \log\left(d^{(2r)^{\nu}}\right) \|h(r)\|\right)\\
& \leq 2^{\nu+1}c\log(d) A \sum_r r^{2\nu} \|h(r)\|. \label{seriesinr}
\end{align} 
In the first step we use the triangle inequality and restrict the summation to terms that contribute a non-zero value. In the second step, we change the order of the summations. In the third step we use Theorem \ref{SIE} and the fact that the support of $h(r)$ only grows like a polynomial in $r$ as stated in Remark \ref{remark1}. We use the polynomial $P(r)=(2r)^{\nu}$, but clearly the specific choice of the polynomial will only influence the constant prefactor. In the last step, we use the assumption on the increase of $M(r)$.
Clearly the condition on the decay of $\|h(r)\|$ assures that the last summation over $r$ converges.
\end{proof}
This proposition is most interesting when the geometry of the bipartition of $\mathcal{L}$ in subsets $\mathcal{B}_1,\mathcal{B}_2$ is not too complicated. Indeed, for complicated bipartitions, the size of the area might be of the same order as the volume, $L^{\nu}$. For more regular bipartitions, like a rectangular subset, the size of the boundary is typically only of the order $L^{\nu-1}$. Moreover, for such cases $M(r)$ typically only grows as $r$.

For strictly local Hamiltonians, or interactions with rapid exponential decay, the previously obtained bounds on the entanglement rate with a polynomial dependence on the dimension of the Hilbert spaces were good enough to obtain a similar result \cite{bravyi2006lieb} and prove that the entanglement rate scales like the area of the boundary of a bipartition. In contrast, it is clear that for general quasi-local Hamiltonians, only the logarithmic dependence of Theorem \ref{SIE} is strong enough to make the summation over $r$ in expression \eqref{seriesinr} converge.

\subsection{Exact Quasi-Adiabatic Continuation}\label{sec:exactQAC}
In this subsection we discuss the formalism of quasi-adiabatic continuation\footnote{The term \textit{quasi-adiabatic} is somewhat of a misnomer since we work with exact filter functions. The original authors \cite{HastingsWen,PhysRevB.69.104431} worked with approximate Gaussian filter functions. Some authors \cite{Auto} prefer the terminology \textit{spectral flow}.}, first introduced by Hastings and Wen in \cite{HastingsWen,PhysRevB.69.104431} and further developed and used by Hastings and collaborators \cite{hastings2007quasi,hastings2009quantization,bravyi2010topological} as well as other authors \cite{osborne2007simulating,Auto}. The aim of this subsection is to present the results that where obtained in the literature and that are needed in the rest of this paper. Since these results are scattered throughout the literature and the conventions and notations of different authors vary, we present a self contained introduction to the subject of quasi-adiabatic continuation in Appendix \ref{appendixa}.

The problem we consider is as follows. We have a quantum spin system defined on a lattice. Consider a path of Hamiltonians $H(s)$ smoothly depending on a parameter $s \in [0,1]$ such that there is a uniform lower bound for the gap $\Delta$ above the ground state energy of these Hamiltonians. We call such an interpolation a quasi-adiabatic path. The rigorous definition is as follows.

\begin{definition}\label{qapath} Consider a quantum spin system defined on a lattice $\mathcal{L}$ and let $\Phi_s$ be LR-local and gapped potentials for all $s \in [0,1]$. We call $\Phi_s$ a \textbf{quasi-adiabatic path} between $\Phi_0$ and $\Phi_1$ if the following conditions are satisfied. The potentials $\Phi_s$ are differentiable with respect to $s$. More specific, we require that $\partial_s \Phi_s(\mathcal{V}) \subset \mathcal{A}_{\mathcal{V}}$ for all finite $\mathcal{V} \subset \mathcal{L}$ and that there exists a constant $C_N$ such that for all $s$, $\|\partial_s \Phi_s(\mathcal{V})\| \leq C_N \|\Phi_s(\mathcal{V})\|$.
Moreover, we demand that the LR-locality is uniform in the sense that there exists a super-polynomial decay function $K$ that dominates the decay functions of all $\Phi(s),\partial_s\Phi_s$. We denote by $\Delta>0$ a uniform lower bound on the gap of the interactions $\Phi(s)$.
\end{definition}

We immediately limit ourselves to translationally invariant systems although all calculations can be done similarly for spatially varying interactions. 
Let us note that the formalism applies to every eigenstate whose corresponding eigenvalue is separated from the rest of the spectrum by a gap, or even every subspace of eigenstates whose eigenvalues are separated from the rest of the spectrum. We restrict our discussion here to gapped unique ground states only. 
Moreover, although we only apply the formalism to Hamiltonians on finite lattices, the formalism of quasi-adiabatic continuation can be rigorously used in the thermodynamic limit \cite{Auto}.

Definition \ref{qapath} induces an equivalence relation on the  gapped, LR-local potentials. Indeed, it is clear that this relation is transitive. Let $\Phi_s$ be a path connecting $\Phi_0,\Phi_1$ and $\tilde{\Phi}_s$ a path connecting $\Phi_1,\Phi_2$. Then, 
\begin{equation}
\Psi(s) = \begin{cases}
   \Phi(2s) & \text{if } 0\leq t\leq 1/2 \\
  \tilde{\Phi}(2s-1) & \text{if } 1/2\leq t\leq 1
  \end{cases}
\end{equation}
is a quasi-adiabatic path connecting $\Phi_0$ and $\Phi_2$.

Let us now fix a finite $\mathcal{V} \subset \mathcal{L}$ and denote the Hamiltonians induced by the potentials $\Phi(s)$ simply by $H(s)$ and their unique and gapped ground states by $\ket{\psi_0(s)}$. The results we obtain are independent of the finite subset $\mathcal{V}$. The evolution of the ground states $\ket{\psi_0(s)}$ can, under general conditions \cite{avron1987adiabatic}, be expressed exactly as  $\ket{\psi_0(s)} = U(s)\ket{\psi_0(0)}$. The unitaries $U(s)$ are the solutions of a differential equation with generator $K(s)$,  
\begin{equation}
\frac{dU(s)}{ds} = iK(s)U(s).
\end{equation}
We are interested in the structure of the generator $K(s)$ of these unitaries. Hastings has shown that these generators are quasi-local Hamiltonians \cite{HastingsWen,PhysRevB.69.104431}. This last statement is highly non-trivial and very powerful.  

To construct the operator $K(s)$, we need a so-called filter function $F(t)$ which is an odd function that decays rapidly in time (faster than any polynomial) and such that its Fourier transform satisfies $\hat{F}(\omega) = -1/\omega$ for $|\omega| \geq \Delta$. That such a function exists is not trivial, but can be proven \cite{graham1981class,vaaler1985some,ingham1934note,dziubanski1998band}. We give an argument in Lemma \ref{fourier} in Appendix \ref{appendixa}.

We now use such a filter function to construct the quasi-adiabatic continuation operator. Notice that we immediately drop the dependence of $K$ on the filter function $F$ in the notation.
The generator of the quasi-adiabatic evolution is given by 
\begin{equation}
K(s) = -i\int_{\mathbb{R}} F(\Delta t) e^{iH_st}\left(\partial_sH_s\right)e^{-iH_st}dt.
\end{equation}
Using perturbation theory we can show that, indeed,
\begin{equation}
iK(s)\ket{\psi_0(s)} =  \partial_s \ket{\psi_0(s)}
\end{equation}
which justifies the definition of $K(s)$. Furthermore, the generator $K(s)$ is a quasi local Hamiltonian. 
More precisely, $K(s)$ can be written as a sum of quasi local terms that decay super polynomially in $r$,
\begin{equation}
K(s) = \sum_{v\in \mathcal{V}}\sum_{r\geq 0} k_v(r),  \quad \supp(h_v(r)) \subset B_r(v), \quad \|h_v(r)\| \leq f(r)
\end{equation}
with $\lim_{r\rightarrow \infty}f(r)P(r) =0$ for every polynomial $P(r)$. For a proof of this statement we refer to Proposition \ref{genlocal} in Appendix \ref{appendixa}.

\subsection{The Stability of the Area Law of the Entanglement Entropy}\label{sec:stabilityarealaw}
The equivalence relation induced by Definition \ref{qapath} on the set of gapped, bounded, LR-local interactions also defines an equivalence relation on the set of ground states of these Hamiltonians. We refer to the equivalence classes as gapped quantum phases.
We now prove that the entanglement entropy relative to a fixed bipartition of a quantum spin system is the same for all states in a given quantum phase, up to a term that scales like the boundary area of the bipartition. Hence, we prove that an area law for one specific ground state automatically carries over to all other ground states that are in the same quantum phase.
\begin{definition}\label{phase}
Let $\mathcal{L}$ be a lattice and $\Phi_s$ a quasi-adiabatic path. Take a finite $\mathcal{V}\subset\mathcal{L}$ and denote the Hamiltonians induced by the potentials $\Phi_s$ on $\mathcal{A}_{\mathcal{V}}$ simply by $H(s)$. Then the unique ground states $\ket{\psi(0)},\ket{\psi(1)}$ of $H(0)$ and $H(1)$ respectively are in the same \textbf{gapped quantum phase}.
\end{definition}

Clearly, the property that the ground states of $H(0),H(1) \in \mathcal{A}_{\mathcal{V}}$ are in the same phase is independent of the set $\mathcal{V}$. Hence, we the above definition can also be applied to the sets of ground states $\{\ket{\psi(0)_{\mathcal{V}}}\},\{\ket{\psi(1)_{\mathcal{V}}}\}$ for all finite $\mathcal{V}\subset\mathcal{L}$. 

To apply the formalism of quasi-adiabatic continuation, we need that the Hamiltonians $H(s)$ are LR-local. Indeed, this is a crucial requirement to prove that the generator of the quasi-adiabatic evolution is quasi-local.  

Let us now define what it means for a quantum spin system to satisfy an area law for the von Neumann entropy.

\begin{definition}\label{arealaw}
Let $\Phi$ be a gapped, translation invariant potential on $\mathbb{Z}^{\nu}$ that generates Hamiltonians $H_L$ on $\mathcal{L}=\mathbb{Z}_L^{\nu}$ for all $L$.
We say that this quantum spin system satisfies an area law if the following holds. Let $S(\mathcal{B})$ be the entanglement entropy of the unique ground state of $H_L$ relative to a bipartition $\mathcal{B}_1,\mathcal{B}_2$ of $\mathcal{L}$, then
\begin{equation}
S(\mathcal{B}) \leq CA(\mathcal{B})
\end{equation}
with $C$ a constant independent of $L$ and the bipartition and $A(\mathcal{B})$ the area of the boundary between $\mathcal{B}_1,\mathcal{B}_2$. 
\end{definition}

The non-trivial part of this definition is the statement that $C$ is independent of $L$. Indeed, for a fixed $L$ we can always take $C=\log(\dim(\mathcal{H}_{\mathcal{L}}))$. Hence, both the definition of a gapped quantum phase and the area law property are best formulated for sequences of states defined on lattices of increasing size. The formulation in terms of a potential of Definitions \ref{phase} and \ref{arealaw} allows us to look at the sequence of ground states of the Hamiltonians $H_L$ generated by this potential.

A similar definition holds for ground state subspaces with a finite degeneracy $q$. Moreover, we are not restricted to lattices $\mathcal{L}=\mathbb{Z}_L^{\nu}$ or translation invariant Hamiltonians, but we should be able to define the system on lattices of increasing size, hence some spatial homogeneity seems necessary.

If we can bound the rate at which entanglement is generated along a quasi-adiabatic path, we can prove an upper bound on the total change of entanglement along the entire path. Recall that $m(v) = \dd(v,\partial \mathcal{B}_1)$ for $v \in \mathcal{B}_2$, $m(v) = \dd(v,\partial \mathcal{B}_2)$ for $v \in \mathcal{B}_1$ and $M(r) = \{v\in \mathcal{L} \:|\: m(v) \leq r \}$.  Moreover, we have that $M(r) \leq 2A(2r)^{\nu}$ for $\mathbb{Z}^{\nu}$ equipped with $\dd$.

\begin{theorem}\label{theoremstab}
Consider the lattice $\mathbb{Z}^{\nu}$ equipped with the metric $\dd$ and let each site support a Hilbert space of dimension $d$. Let $\Phi_s$ be a quasi-adiabatic path on the quantum spin system. Consider the finite lattice $\mathcal{L}=\mathbb{Z}_L^{\nu}$ and denote the Hamiltonians induced by the potentials $\Phi(s)$ simply by $H(s)$.  Denote by $\ket{\psi(0)},\ket{\psi(1)} \in \mathcal{H}_{\mathcal{L}}$ the unique ground states of $H(0)$ and $H(1)$ respectively. Let $\mathcal{B}_1,\mathcal{B}_2$ be a fixed bipartition of the lattice. Denote the size of the area of the boundary between $\mathcal{B}_1$ and $\mathcal{B}_2$ by $A$.
Then, the entanglement entropy of $\ket{\psi(0)}$ and $\ket{\psi(1)}$ differ at most by a constant times the area of the boundary between $\mathcal{B}_1,\mathcal{B}_2$. Therefore, if $\ket{\psi(0)}$ satisfies an area law, so does $\ket{\psi(1)}$ and vice versa.
\end{theorem}
\begin{proof}
We are interested in the entanglement entropy of the ground states $\ket{\psi(s)}$ of $H(s)$ and more precisely in the rate of change of this quantity as $s$ changes. We can bound the entanglement rate along $s$ the same way we did in subsection \ref{sec:entgeneration},
\begin{align}
\left|\frac{dS_{\mathcal{B}_1}(s)}{ds}\right| &= \left|\Tr\left(K(s)[\ket{\psi(s)}\bra{\psi(s)},\log \rho_{\mathcal{B}_1} \otimes \mathds{1}_{\mathcal{B}_2}]\right)\right| \label{rates1}\\
& \leq \sum_{r\in \mathbb{N}}\sum_{v\in\mathcal{L}:m(v)\leq r}  \left|\Tr\left(k_{v,r}(s)[\ket{\psi(s))}\bra{\psi(s)},\log \rho_{\mathcal{B}_1} \otimes \mathds{1}_{\mathcal{B}_2}]\right)\right|\\
&\leq \sum_{r\in \mathbb{N}}M(r)\log\left(d^{(2r)^{\nu}}\right)\|k_r(s)\|\\
&\leq cA2^{\nu+1}\log d\sum_{r\in \mathbb{N}}r^{2\nu}\|k_r(s)\| \label{rates2}.
\end{align} 

Since $\|k_r(s)\|$ decays superpolynomially, it is clear that this last sum is bounded by a constant. Hence we find that the rate of change is bounded by a constant $C$ times the area $A$ of the boundary  between $\mathcal{B}_1,\mathcal{B}_2$, 
\begin{equation}\label{arearate3}
\left|\frac{dS_{\mathcal{B}_1}(s)}{ds}\right|  \leq CA_.
\end{equation}
Now consider two Hamiltonians $H(0),H(1)$ which are in the same quantum phase. By definition \ref{phase}, there exists a quasi-adiabatic path connecting them. We can bound the entanglement rate along this path and upon integration of \eqref{arearate3} we find that
\begin{equation}
\Delta S_{\mathcal{B}_1} = S_{\mathcal{B}_1}(1) - S_{\mathcal{B}_1}(0) \leq CA
\end{equation}
for a constant $C$ independent of the system size $L$ or boundary area $A$. Hence, we have shown that all ground states within the same gapped quantum phase have the same area law behaviour. Either they all satisfy the area law for the entanglement entropy or they all violate it.
\end{proof}
The proof can be generalized to quantum spin systems defined on different lattices that satisfy the condition in Remark \ref{remark1}, which implies that $M(r)$ is bounded by a polynomial in $r$.

\subsection{Degenerate Ground States}\label{degenerategs}
The stability of the area law readily generalizes  to the case of a finitely degenerate ground state subspace. The formalism of quasi-adiabatic continuation still applies to these systems. Let us take two Hamiltonians $H(0),H(1)$ that are in the same phase. Necessarily, the ground state degeneracy of $H(0)$ and $H(1)$ is the same. Let us assume that there is a basis of the ground state subspace of $H(0)$ such that all basis vectors satisfy an area law. Under quasi adiabatic evolution, this basis is mapped to a basis of the ground state subspace of $H(1)$. By Theorem \ref{theoremstab}, these basis states all satisfy the area law.

Moreover, given a basis of the ground state subspace of a Hamiltonian such that all basis vectors satisfy the area law, we can bound the entanglement of a general state in the ground state subspace. Indeed, a finite superposition of states that satisfies the area law, still satisfies the area law itself \cite{linden2006entanglement}. Given a two orthogonal ground states $\ket{\psi_1},\ket{\psi_2}$, we find that
\begin{equation}\label{linden}
S(\alpha\ket{\psi_1}+\beta\ket{\psi_2})\leq 2\left(|\alpha|^2S(\ket{\psi_1})+|\beta|^2S(\ket{\psi_2})+\mathsf{h}(|\alpha|^2,|\beta^2|)\right)
\end{equation}
For a general basis or for larger superpositions, the entanglement can still be bounded and the area laws continue to hold. However, the expression for the increase of the prefactors is more complicated \cite{linden2006entanglement} than equation \eqref{linden}.

For a degeneracy that grows with the system size, the existence of linearly independent ground states of $H_0$ that satisfy the area law still implies the existence of an equal number of linearly  independent ground states of $H_1$ with the same property. However, one can draw no conclusions about the entanglement of a general ground state in the huge ground state subspace of such a Hamiltonian.

\subsection{Fermionic Lattice Systems}\label{sec:fermions}
We can extend Theorem \ref{theoremstab} to systems consisting of fermions living on a lattice and with the fermionic Hamiltonian $H_f$ local in the sense of fermionic modes. An example of such a fermionic Hamiltonian is given by the famous Fermi-Hubbard model,
\begin{equation}
H_{FH} = -t \sum_{\braket{ij}\sigma}c^{\dagger}_{i\sigma}c_{j\sigma} + U\sum_i n_{i\uparrow}n_{i\downarrow}-\mu\sum_{i}n_i.
\end{equation}
The modes correspond to the lattice points where the fermions live and possibly extra labels such as the spin of the fermions. For simplicity we only consider the lattice labels and ignore these extra degrees of freedom. Fix an ordering of the fermions and apply the Jordan-Wigner transformation \cite{jordanwigner}. This transformation turns the fermionic Hamiltonian into a spin Hamiltonian of the following form,
\begin{equation}
H_{JW}= \sum_{v\in \mathcal{L}} h_v\otimes Z_{G_v}.
\end{equation}
Here, $h_v$ is a local interaction term centred around site $v$ and $Z_{G_v}$ is a non local string of pauli $Z$ operators working on a certain region $G_v$. This region depends on the chosen ordering.

We first note the following simple fact.
\begin{proposition}\label{invarianceLU}
The maximal entanglement rate of a Hamiltonian $H$ relative to a bipartition $\mathcal{B}_1,\mathcal{B}_2$ does not change under unitary transformations of the form $U=U_{\mathcal{B}_1} \otimes U_{\mathcal{B}_2}$,
\begin{equation}
\Gamma(H) = \Gamma(UHU^{\dagger}).
\end{equation}
\end{proposition}

\begin{proof}
It is straightforward to check the following relation for every state $\ket{\psi}$,
\begin{equation}
\Gamma(H,\psi) = \Gamma(UHU^{\dagger},U\psi).
\end{equation}
Indeed, due to the factorised form of $U$ all unitaries in equation \eqref{entrate} cancel. This implies that the entanglement rate $\Gamma$, which is a maximum over all possible states $\ket{\psi}$ is equal for both Hamiltonians. 
\end{proof}
Several possible ways to quantify the entanglement of a fermionic system have been studied in the literature \cite{banuls2007entanglement}. Here we define the entanglement rate of $H_f$ as that of $H_{JW}$. It now suffices to bound the entanglement rate of a term $h_v\otimes Z_{G_v}$ similarly as we bounded $h_v$.\\
\\
We write $Z_{G_v} = Z^{\mathcal{B}_1}_{G_v} \otimes Z^{\mathcal{B}_2}_{G_v}$ with both operators supported strictly on one side of the bipartition. Since $Z^{\mathcal{B}_1}_{G_v}$ has a spectrum containing an equal number of $\pm 1$, there exists a unitary operator $U_{\mathcal{B}_1}$ such that
\begin{equation}
U_{\mathcal{B}_1} Z^{\mathcal{B}_1}_{G_v} U_{\mathcal{B}_1}^{\dagger} = Z_{v^*} \otimes \mathds{1}_{G_v \setminus \{v^*\}}.
\end{equation}
Here, $v^* \in G_v$ is a vertex site neighbouring the support of $h_v$.
A similar unitary $U_{\mathcal{B}_1}$ can be found for $Z^{\mathcal{B}_1}_{G_i}$. Both these unitaries map the entire string of $Z$ operators to a $Z$ on a single site. We now use Proposition \ref{invarianceLU} to obtain that
\begin{equation}
\begin{aligned}
\Gamma(h_i\otimes Z_{G_i}) &= \Gamma\left(U_{\mathcal{B}_1}\otimes U_{\mathcal{B}_2} (h_i\otimes Z_{G_i}) U_{\mathcal{B}_1}^{\dagger}\otimes U_{\mathcal{B}_2}^{\dagger}\right)\\
&= \Gamma\left(Z_{i_L^*} \otimes h_i \otimes Z_{i_R^*}\right) \leq c\|h_i\| \log(D+2).
\end{aligned}
\end{equation}
\begin{corollary}
The entanglement rate of a fermionic Hamiltonian $H_f$ that is local in the sense of modes, obeys an area law.
\end{corollary}
\begin{proof}
The result follows from the previous discussion and the results on spin lattices and local spin Hamiltonians given in Proposition \ref{entraterealtime}.
\end{proof}
\subsection{The Area Law under Adiabatic Growing}\label{sec:AdGrow}
In the previous subsections we have given a rigorous proof of the stability of the area law in a gapped quantum phase. One of the main motivations of the authors to obtain the stability of the area law was to find a new method to prove the area law itself. This subsection is of a more speculative character and we do not claim that all our arguments can be made rigorous. 

First we give some intuition as to why the stability of the area law can be used to infer properties about the validity of the area law. We give some conditions under which this intuition can be turned into a rigorous proof. Second, we discuss some related work by other authors and some directions to improve on these results.

Consider a two dimensional square lattice $\mathcal{L} \subset \mathbb{Z}^2$ of arbitrary, but finite, size $L \times L$. In this subsection we do not assume periodic boundary conditions. Each vertex $v$ has a local Hilbert space of dimension $d$. We consider a uniformly bounded, local potential $\Phi$ on $\mathbb{Z}^2$ on this lattice such that the Hamiltonians $H_L$ have a gap and a unique ground state. We are interested in the entanglement entropy of the reduced density matrix of the ground state of a square region $\mathcal{V}$ of size $\ell \times \ell$, with $1\ll \ell \ll L$.
Consider the ground state of the Hamiltonian $H_{\mathcal{V}}$ on this $\ell\times \ell$ subset $\mathcal{V}$ and think of it as consisting of the ground state $\ket{\psi_{A}}$ on the $\ell\times \ell$ lattice and a product state of all $\ket{0}$ on all other sites. Now consider the same model, but on a slightly bigger subset $\mathcal{W}:= \mathcal{V}\cup \{v\}$. Here $v \in \mathcal{L} \setminus \mathcal{V}$ is just a single lattice site neighbouring $\mathcal{V}$. We denote the ground state of $H_{\mathcal{W}}$ by $\ket{\psi_{\mathcal{W}}}$.

Intuitively, it should be possible to go from $\ket{\psi_{\mathcal{V}}}\otimes\ket{0}_w$ to $\ket{\psi_{\mathcal{W}}}$ by acting with a unitary that acts only on the state $\ket{0}$ of the extra added site $v$ and the sites close to it, within a ball of radius $\xi$, the correlation length. This formalises the idea that both ground states should be in the same gapped quantum phase once the system is big enough. We keep repeating this procedure until we end up with the ground state of the model on a lattice of size $L \times L$. This procedure is illustrated in figure \ref{figurearea}.

We can now look at the entanglement of this state, relative to a bipartition of the system into the subset $\mathcal{V}$ and the rest. Since we started with a product state between both systems, we started without entanglement. Every step involved a unitary which created an amount of entanglement $\sim \xi^2\log d$. After going from the $\ell\times\ell$ subsystem to an $(\ell+\xi)\times(\ell+\xi)$ subsystem, we applied $\sim 4\ell\xi$ of these unitaries. Hence the state now has no more entanglement than a constant times $\ell$, it has an area law. We can now grow the system until we obtain the ground state on the entire lattice. The unitaries we now apply do not act on the subset $A$, hence they create no additional entanglement. See again figure \ref{figurearea} for an illustration of this argument. We conclude that the ground state of our model obeys the area law,
$$
S(\rho_A) \leq C \xi^3\log d \ell.
$$

\begin{figure}
  \centering
  \includegraphics[width=6in]{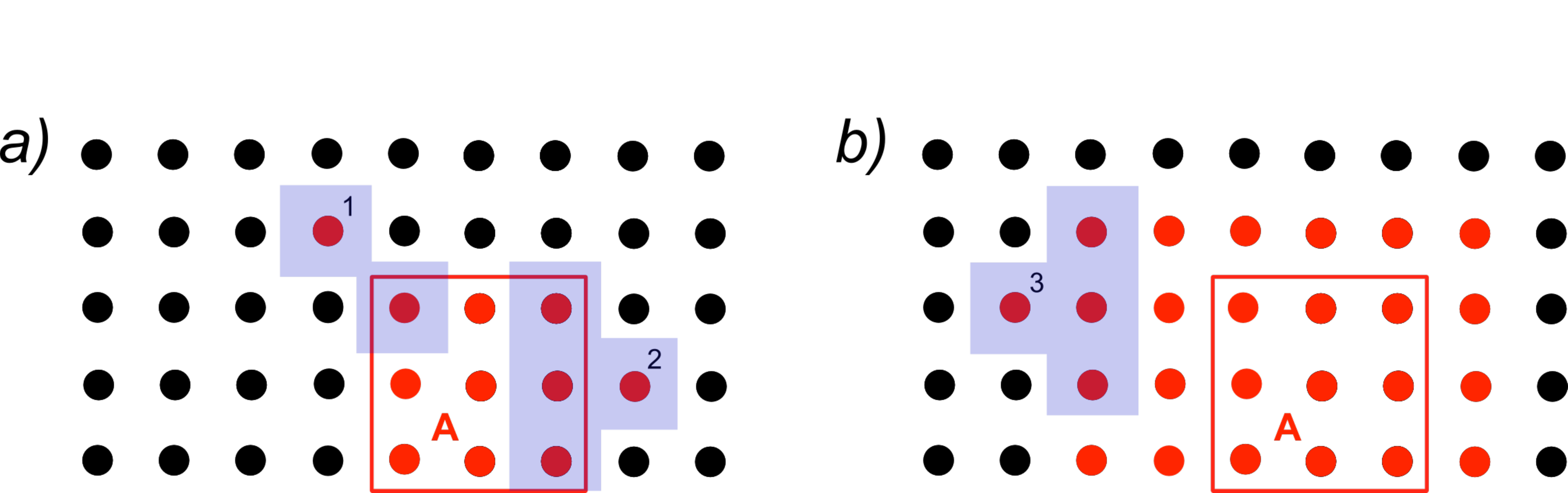}
  \caption[Adiabatic growing]
   {We are interested in the entropy of region $A$. The red dots are in the ground state, the black dots are ancillas in a product state. In figure $a$ we act with two unitaries, blue coloured, to create the ground state on region $A$ and qubits 1,2. After acting with a number of these unitaries proportional to the boundary we are in the situation of figure $b$. Adding more qubits does not change the entropy of region $A$ since the unitary needed to create the ground state on the extra qubit 3 only acts within the complement of $A$.}
  \label{figurearea} 
\end{figure}

There are three obstacles that need to be overcome to turn this argument into a proof. First, the unitaries that are used to grow the ground state can never be expected to act strictly on a number of sites $\sim \xi^2$. In general they are generated by a quasi-local Hamiltonian which is more or less supported around these sites. This situation is of course exactly the kind of problem that Theorem  \ref{SIE} was designed to solve.

There are two more serious issues that we cannot solve without extra assumptions. First, we ignored the fact that typically Hamiltonians have edge modes. A model that has a unique ground state on a closed lattice, like a sphere, or in the thermodynamic limit, can have a degenerate ground state subspace on a lattice with open boundary conditions. It will have a, possibly exponentially large, degenerate ground state subspace, with exponentially small splitting, separated from the rest of the spectrum by a gap $\Delta$. This problem can be dealt with if we can make the ground state unique by introducing some local boundary Hamiltonian. It is known that this is possible for several interesting models like the toric code or more general string net models \cite{levin2013protected,kitaev2012models,2014arXiv1402.4081L}, but there are certainly counterexamples like the fractional quantum Hall effect \cite{kitaev2006anyons,wen1995topological}.

The last problem concerns the existence of a quasi-local unitary that maps the ground state of a model on $\ell \times \ell$ sites to the ground state on a slightly bigger lattice. This assumption expresses the fact that if $\ell$ is large enough and the model has a well defined thermodynamic limit, adding a single particle shouldn't change anything concerning the phase of the ground state. Hence, there should be a gapped adiabatic path between both states, which ensures the existence of such a unitary. Unfortunately, counterexamples to this condition exist, for instance Haah's cubic code does not satisfy this requirement, although it obeys the area law \cite{haah2011local}.

If both conditions hold, a proof for the area law follows immediately from the following corollary.
\begin{corollary} \label{corgrow}
Consider the lattice $\mathbb{Z}^{\nu}$ with the metric $\dd$ and let every site support a Hilbert space of dimension $d$. Let $\{\ket{\psi_{\mathcal{V}}}\:|\: \mathcal{V} \subset \mathbb{Z}^{\nu}\}$ be a collection of states, with $\ket{\psi_{\mathcal{V}}} \in \mathcal{H}_{\mathcal{V}}$ for every finite $\mathcal{V} \subset \mathbb{Z}^{\nu}$. Take such a finite set $\mathcal{V}$ and let $v \in \mathbb{Z}^{\nu} \setminus \mathcal{V}$ be a site not in $\mathcal{V}$. Suppose that for all such $\mathcal{V},v$ the state $\ket{\psi_{\mathcal{V}\cup \{v\}}}$ can be obtained from $\ket{\psi_{\mathcal{V}}}\otimes\ket{0}$ by evolving with a quasi-local Hamiltonian centred around $v$ for finite time $T$. Moreover, we require an uniform upper bounds $T$ on the time and the existence of a super-polynomial decay function $K$ that can be used for all Hamiltonians.  Then, this collection of states satisfies an area law in the sense that there is $C$ independent of $\mathcal{V}$ such that for every bipartition $\mathcal{B}_1,\mathcal{B}_2$ of $\mathcal{V}$,
$$
S(\Tr_{\mathcal{B}_1}(\ket{\psi_{\mathcal{V}}}\bra{\psi_{\mathcal{V}}})) \leq CA
$$
with $A$ the size of the boundary between $\mathcal{B}_1$ and $\mathcal{B}_2$.
\end{corollary}

\begin{proof}
Take a finite $\mathcal{V}$ and a bipartition $\mathcal{B}_1,\mathcal{B}_2$. Consider the state $\ket{\psi_{\mathcal{B}_1}}$. We then go to the state $\ket{\psi_{\mathcal{V}}}$ by adding more and more ancilla sites and evolving with quasi-local Hamiltonians. We can sum all the contributions to the entanglement when we evolve $\ket{\psi_{\mathcal{B}_1}} \otimes \ket{0}^{|\mathcal{B}_2|}$ to $\ket{\psi_{\mathcal{V}}}$ with these Hamiltonians. This gives us expressions very similar to the one used in the proof of Theorem \ref{theoremstab}. These expressions can be bounded analogously as inequalities \eqref{rates1}-\eqref{rates2} and the claim follows. Moreover, it is clear that there exists a constant $C$ that only depends on the uniform quantities $T,K$.
\end{proof}

\begin{remark}
Corollary \ref{corgrow} can be used in the following situation. 
Consider again the lattice $\mathbb{Z}^{\nu}$ and suppose we have a potential $\Phi$ such that the Hamiltonians $H_L$ on $[0,L]^{\nu}$ have edge modes, such that the degeneracy of the ground state subspace potentially depends on $L$. If we can obtain a ground state of $H_L$ from a ground state on a smaller lattice by evolving with a quasi-local Hamiltonian centred around the boundary, we can still prove that this collection of ground states satisfies the area law in the sense of Corollary \ref{corgrow}.

One way to obtain such a collection is by using boundary terms. Fix an $R>0$, we call an operator $B$ a boundary term if it acts only on sites $v\in\mathcal{L}$ whose distance to the boundary between $\mathcal{V}$ and $\mathbb{Z}^{\nu} \setminus \mathcal{V}$ is smaller than $R$. Now suppose that for all finite subsets $\mathcal{V} \subset \mathbb{Z}^{\nu}$,
\begin{itemize}
\item[(a)] There exists local boundary terms $B$ that can be added to $H_{\mathcal{V}}$ such that  the Hamiltonian $\tilde{H}_{\mathcal{V}} = H_{\mathcal{V}}+B$ is gapped and has a unique ground state.
\item[(b)] One can go from the ground state of $\tilde{H}_{\mathcal{V}}$ to that of $\tilde{H}_{\mathcal{V} \cup v}$ by evolving with a quasi-local Hamiltonian centred around $v$. Here $v \in \mathbb{Z}^{\nu} \setminus \mathcal{L}$ is a site neighbouring $\mathcal{V}$.
\end{itemize}
Then the collection of unique ground states of $\tilde{H}_{[0,L]^{\nu}}$ for all $L$ satisfies the area law by Corollary \ref{corgrow}.
\end{remark}

Similarly as Theorem \ref{theoremstab}, we can generalize this corollary to different lattices. By now, it is clear which conditions we need to obtain the result.

\begin{remark}
We would like to comment on a subtle point. Since the procedure we described above consists of the application of several quasi-local unitaries, one could naively think that the state we start with and the state we end up with are in the same phase. However, this is not necessarily true, since the length of the adiabatic path we need scales with the system size, in contrast to Definition \ref{phase}. Indeed it is well known that non trivial topologically ordered models such as the toric code and other string net models can be obtained by a very similar procedure as the one we described \cite{dennis2002topological,konig2009exact}. We only assume that the ground states on lattices of size $N$ or $N+1$ are in the same phase and this does not necessarily imply  that ground states on $N$ and $N^2$ sites are in the same phase.
\end{remark}
Independently, very similar ideas were reported and further elaborated on by other authors. For completeness and to illustrate the general idea of turning the stability of the area law in a proof of the area law itself, we summarize the work of these authors.

In \cite{cho2014entanglement} an area law was proven for systems obeying the following two conditions. First, there exists a sequence of Hamiltonians 
$$
\left\{H_1,H_2,\ldots, H_N\right\}
$$
acting on $N$ qubits, such that the next Hamiltonian is constructed from its predecessor by adding a term only at the boundary close to the added point. Furthermore all Hamiltonians have a gap at least $\Delta$. This requirement takes care of the problem concerning edge modes that was mentioned before.

Second, it is assumed that the ground states $\ket{\psi_0^{k}},\ket{\psi_0^{k+1}}$ of two consecutive Hamiltonians have a finite, non zero overlap. Using this condition the author can prove the existence of a gapped Hamiltonian path between consecutive Hamiltonians. Hence the formalism of quasi-adiabatic continuation assures the existence of a quasi-local unitary that maps one ground state to its successor. This takes care of the second problem we mentioned. Important, it gives a clear condition under which such a path can be rigorously proven.

Under these conditions, the following result can be proven with the SIE theorem, since the entropy created by such quasi-local unitaries can be bounded. 
\begin{theorem}\cite{cho2014entanglement}
Consider a $D$-dimensional spin system satisfying the above conditions. The entanglement entropy of a ball of radius $R_0$ is bounded by
\begin{equation}
S(\rho_{\text{ball}}) \leq c_{D-1}R_0^{D-1}+c_{D-2}R_0^{D-2}+\ldots + c_1R_0+c_0.
\end{equation}
The constants $c_i$ are system dependent constants.
\end{theorem}

Recently, related ideas were also reported in \cite{swingle2014renormalization}. Inspired by renormalization group ideas, the authors proposed the following definition which is conjectured to capture physically relevant gapped quantum phases.
\begin{definition}\cite{swingle2014renormalization}
A $D$ dimensional $s$ source RG fixed point is a system whose ground state on $(2L)^D$ sites can be constructed from $s$ copies of the ground state on $L^D$ sites plus some unentangled degrees of freedom by acting with a quasi-local unitary on these spins.
\end{definition}

This definition can be modified to include the use of models on $L^D$ sites, different than the original model, in the construction of the model on $(2L)^D$ sites. This is for instance the case in Haah's cubic code \cite{haah2014bifurcation}. The entropy generated by the  quasi-local unitaries can be bounded by the SIE theorem. If $s$ is not too big, $s< 2^{D-1}$, repeating this procedure gives an area law. The calculation of the scaling of the entanglement entropy of such systems is similar to the one for branching MERA \cite{evenbly2013scaling}. In this calculation, only strictly local unitaries are considered. Hence, the main difference here is the presence of quasi-local unitaries. This additional issue is solved by the SIE theorem.
\begin{remark}
A very similar definition was given in \cite{zeng2014stochastic}. The authors called the systems they studied \textit{gapped quantum liquids}. They corresponds intuitively to the $s=1$ source RG fixed points. In \cite{swingle2014renormalization} these systems were called \textit{topological quantum liquids}. For all these systems, the growing procedure and the SIE theorem imply the validity of the area law.
\end{remark}
\section{Conclusion}
In this paper we addressed two main subjects. First, we discussed the entanglement that can be created by a local Hamiltonian in a quantum spin system. We gave a comprehensive overview of the motivation and previous work. The original contribution is this paper is the solution of the dynamical part of this question. We proved a sharp upper bound on the maximal instantaneous rate at which a local Hamiltonian can create entanglement in a spin system even in the presence of arbitrary ancillas. 

Second, the upper bound on the entanglement rate allowed us to prove the stability of the area law for the entanglement entropy in gapped quantum phases. An area law for a ground state of a local gapped Hamiltonian automatically carries over to all systems to which it is connected via a gapped path of Hamiltonians. The formalism of quasi-adiabatic continuation provides the existence of quasi-local Hamiltonians which governs the evolution of a ground state to the ones connected with it by such a gapped path of Hamiltonians. The entanglement created by this evolution can be controlled and shown to be proportional to the area of boundary of the bipartition and not to the volume of the constituents of the bipartition. This result carries over to systems with a finitely degenerate ground state subspace and fermionic lattice systems. We also discussed how under certain assumptions a similar argument can be used to prove the area law itself. 

Several open questions remain. First, it would be interesting to see if a similar bound can be obtained for the mutual information. Given a local Lindblad generator, can we prove a non trivial upper bound on the speed at which the mutual information can change? 

Another open problem concerns the area law itself. As we discussed, several authors have already shown that under certain restriction our results can be used to prove the area law. A natural question is how far these results can be generalised. Moreover, due to the stability result it suffices to show the area law for a single system in every phase. For a very large class of physical systems, we conjecture that there is a commuting, frustration free Hamiltonian in the same phase, this system corresponds to the renormalization fixed point. Such a system trivially fulfils the area law. Alternatively, one could show that certain phases have at least one representative ground state that can be written as a PEPS \cite{verstraete2008matrix}. Such a state satisfies the area law by construction.

Third, a natural generalization of this work concerns the Renyi entropies. It would be interesting to understand the different behaviour of the entanglement rate as measured by these entropies or by the von Neumann entropy.

\section*{Acknowledgements}
We would like to acknowledge I. Cirac, B. De Vylder, J. Haegeman, E. Lieb, S. Michalakis, T. Osborne and  B. Swingle for extremely inspiring and useful discussions. This work is supported by an Odysseus grant from the FWO, the FWF grants FoQuS and Vicom, and the ERC grant QUERG. 

\bibliographystyle{abbrv}
\bibliography{bibarealaw}

\section*{Appendix A}\label{appendixa}
For convenience of the reader who is not familiar with the subject of quasi-adiabatic continuation, we give a self contained introduction in this appendix. The original papers where the formalism was presented and further developed are among others \cite{HastingsWen,PhysRevB.69.104431,hastings2007quasi,hastings2009quantization,bravyi2010topological}.
Good overviews are also given in \cite{hastings2010locality,osborne2007simulating,Auto}, on which the material in this appendix is based. The aim of this Appendix is to give a short introduction, that presents all the necessary definitions and results, which are scattered throughout the literature.

We again consider the situation discussed before and in Definition \ref{qapath}. For notational convenience we restrict ourselves to strictly local, translation invariant Hamiltonians $H(s)$.  We can modify Definition \ref{qapath} to obtain the following definition that applies to these restricted interactions. We emphasise that the restriction in Definition \ref{qapath2} is only used for notational convenience and that all relevant results are also valid for the quasi-adiabatic paths defined in Definition \ref{qapath}.

\begin{definition}\label{qapath2}
Consider a quantum spin system defined on a lattice $\mathcal{L}$ and let $\Phi_s$ be local, bounded and gapped potentials for all $s \in [0,1]$. We call $\Phi_s$ a \textbf{quasi-adiabatic path} between $\Phi_0$ and $\Phi_1$ if the following conditions are satisfied. The potentials $\Phi_s$ are differentiable with respect to $s$. More specific, we require that $\partial_s \Phi_s(\mathcal{V}) \subset \mathcal{A}_{\mathcal{V}}$ for all finite $\mathcal{V} \subset \mathcal{L}$.
Moreover, we require the existence of constant $R,C_N$,
 \begin{align}
&\Phi_s(\mathcal{V}),\partial_s\Phi_s(\mathcal{V})=0 \text{ if } \text{diam}(\mathcal{V})>R\\
&\sup_s\sup_{\mathcal{V}}\{\|\Phi_s(\mathcal{V})\|,\|\partial_s\Phi(\mathcal{V})\|\}  \leq C_N.
\end{align}
We denote by $\Delta>0$ a uniform lower bound on the gap of the interactions $\Phi(s)$.
\end{definition}

Let us now fix a finite $\mathcal{V} \subset \mathcal{L}$ and denote the Hamiltonians induced by the potentials $\Phi(s)$ simply by $H(s)$ and their unique and gapped ground states by $\ket{\psi_0(s)}$. The results we obtain are independent of the finite subset $\mathcal{V}$. We mentioned in subsection \ref{sec:exactQAC} of the main text that evolution of the ground state $\ket{\psi_0(s)}$ can be expressed  as  $\ket{\psi_0(s)} = U(s)\ket{\psi_0(0)}$ and the unitaries $U(s)$ are the solutions of a differential equation with generator $K(s)$,  
\begin{equation}
\frac{dU(s)}{ds} = iK(s)U(s).
\end{equation}
Our main goal is to show that the generators $K(s)$ are quasi-local Hamiltonians.  We first show the existence of the filter functions that we used in section \ref{sec:exactQAC}. We start with the following result from Fourier analysis.

\begin{lemma}\label{fourier} Let $\Delta > 0$.
There exists an odd function $F$ such that $F(t)$ decays super-polynomially in time and such that $\hat{F}(\omega) = -1/\omega$ for $|\omega| \geq \Delta$. Here $\hat{F}$ is the Fourier transform of the function $F$.
\end{lemma}
\begin{proof}
We follow the argument given in \cite{hastings2010locality,hastings2010quasi}. An explicit example of such a function was given in \cite{Auto}. From now on, we assume that $\Delta = 1$. We start with a function $g$ such that its Fourier transform $\hat{g}$ has compact support $[-1,1]$, $\hat{g}(0)=1$ and $g$ itself vanishes rapidly. It is a well-known result in Fourier theory that such functions exist, several different arguments are used in the literature \cite{ingham1934note,dziubanski1998band,vaaler1985some,graham1981class}. In \cite{ingham1934note}, functions are constructed such that $g$ decays faster than $\exp(-|t|\epsilon(|t|))$ for large $t$. Here, $\epsilon$ can be any monotonically decreasing positive function with 
\begin{equation}
\int_1^{\infty} \frac{\epsilon(y)}{y}dy \leq \infty.
\end{equation}
Take such a $g$ even. We can now define $f(t) = \delta(t)-g(t)$, which is also an even function. Moreover, $\hat{f}(0) = 0$ and $\hat{f}(\omega) = 1$ for $|\omega| \geq 1$. Finally, we can build the desired function by a convolution with the sign function,
\begin{equation}
F(t) = \frac{i}{2} \int_{\mathbb{R}} duf(u)\text{sign}(t-u).
\end{equation}
Since  $|F(t)|\leq \left|\int_{|t|}^{\infty}f(u)du\right|$, the super-polynomial decay of $f$ implies a similar large $t$ behaviour of $F$. Furthermore, we have that
\begin{equation}
\hat{F}(\omega) = \frac{i}{2}\int_{\mathbb{R}}dt\exp(i\omega t)\int_{\mathbb{R}}du f(u)\text{sign}(t-u).
\end{equation}
This last expression can be integrated by parts in $t$ to yield
\begin{align}
\hat{F}(\omega) &= \text{boundary terms}-\frac{1}{\omega}\frac{1}{2}\int_{\mathbb{R}}d\left(\int_{\mathbb{R}}duf(u)\text{sign}(t-u)\right)e^{i\omega t}\\
&= -\frac{1}{\omega}\int_{\mathbb{R}}duf(u)\delta(t-u)\exp(i\omega t)\\
&= -\frac{1}{\omega}\hat{f}(\omega)
\end{align}
as the boundary terms cancel.
\end{proof}
We now use such functions to define the quasi-adiabatic continuation operator. Notice that we immediately drop the dependence of $K$ on $F$.
\begin{definition} The generator of the quasi-adiabatic evolution is defined as
\begin{equation}
K(s) = -i\int_{\mathbb{R}} F(\Delta t) e^{iH_st}\left(\partial_sH_s\right)e^{-iH_st}dt.
\end{equation}
\end{definition}
To show that this is a good definition we proceed with the following calculation. We have that
\begin{align}
iK(s)\ket{\psi_0(s)} &= \int_{\mathbb{R}} F(\Delta t) e^{iH_st}\partial_sH_se^{-iH_st}\ket{\psi_0(s)}dt\\
&= (\mathds{1}-P_0(s)) \int_{\mathbb{R}} F(\Delta t) e^{iH_st}\partial_sH_se^{-iH_st}\ket{\psi_0(s)}dt\\
&= \sum_{i \neq 0} \ket{\psi_i(s)}\braket{\psi_i(s),\partial_s H_s \psi_0(s)} \int_{\mathbb{R}}F(\Delta t) e^{i(E_i(s)-E_0(s))t}dt\\
&= \sum_{i\neq 0} \frac{1}{E_0(s)-E_i(s)} \ket{\psi_i(s)}\braket{\psi_i(s),\partial_s H_s \psi_0(s)}\\
&= \partial_s \ket{\psi_0(s)}.
\end{align}

We now study the generator $K(s)$. We decompose the local Hamiltonians $H(s)$ as 
\begin{equation}
H(s) = \sum_{v\in \mathcal{V}} h_{\mathbf{j}}(s).
\end{equation}
We now show that $K(s)$ is a quasi-local operator. To lighten the notation we write the quasi-adiabatic evolution of every operator $X$ as
\begin{equation}
\mathcal{F}_s(X) := -i\int_{\mathbb{R}} F(\Delta t) e^{iH(s)t}Xe^{-iH(s)}dt.
\end{equation}
We also need a local approximation of the operator $\mathcal{F}_s(X)$, only supported on a subset $\Lambda \cup \text{supp}(X)$,
\begin{equation}\label{approxquasiadiab}
\mathcal{F}^{\Lambda}_s(X) := -i\int_{\mathbb{R}} F(\Delta t) e^{iH_{\Lambda}(s)t}Xe^{-iH_{\Lambda}(s)}dt.
\end{equation}
It is now clear that the ground state $\ket{\psi_0(s)}$ evolves according to the unitary dynamics generated by the Hamiltonian
\begin{equation}
K(s) = \sum_{v\in \mathcal{V}} \mathcal{F}_s(\partial_s h_{v}) := \sum_{v\in \mathcal{V}} k_{v}(s),
\end{equation}
with $k_{v}(s) := \mathcal{F}_s(\partial_s h_{v})$. We now take an arbitrary origin and for convenience we drop the index referring to the origin both for the interaction $h(s)$ and for $k(s) := \mathcal{F}_s(\partial_s h)$. Given that we consider a translation invariant system, we can just focus on these terms in the subsequent arguments.\\
\\
Our goal is to show that $k(s)$ is a quasi-local operator. By definition this means  we can decompose
\begin{equation}
k(s) = \sum_{r=0}^{\infty} k_r(s)
\end{equation}
such that $k_r(s)$ has growing support but its norm decays superpolynomially in $r$. To obtain such a decomposition we first define the sets
\begin{equation}
\Lambda_r = \{v\left.\right| d(\mathbf{0},v) \leq r\} \bigcup_{s \in [0,1]}\text{supp}(h(s)).
\end{equation}
We will show that $k(s)$ decomposes in local terms $k_r(s)$ such that 
$\text{supp}(k_r(s)) \subset \Lambda_r $.
The main idea is to write $k(s)$ as a telescoping sum of strictly local terms using the approximate evolution defined in \eqref{approxquasiadiab}. We first define
\begin{equation}
k_0(s) = \mathcal{F}^{\Lambda_0}_s(\partial_s h(s))
\end{equation}
and 
\begin{equation}\label{decayingterms}
k_r(s) = \mathcal{F}^{\Lambda_r}(\partial_s h(s)) - \mathcal{F}^{\Lambda_{r-1}}(\partial_s h(s)),\quad r>0.
\end{equation}
With this definition it is clear that $\text{supp}(k_r(s)) = \Lambda_r$ and $k(s) = \sum_{r=0}^{\infty}k_r(s)$. Hence we only need to show that the norm of these terms decays sufficiently rapid in $r$.\\
\\
To show the decay of the operators $k_r(s)$, we use the Lieb-Robinson bounds \cite{liebrobinson,hastings2006spectral,nachtergaele2006lieb} of the original, physical Hamiltonians $H(s)$, see Theorem \ref{theoremliebrobinson}. The calculations generalise to quasi-local interactions $h$ which satisfy a Lieb-Robinson bound, see Theorem \ref{theoremLB2}.

We use the Lieb-Robinson bound to show a statement that is very similar is spirit. For large $r$ the evolution of a local operator $A$ by $H_{\Lambda_r}$ or $H_{\Lambda_{r-1}}$ is almost the same. This is exactly what we need to bound the integrand of the terms \eqref{decayingterms}. The argument is similar to existing ones in the literature \cite{osborne2007simulating,bratteli1981operator}. Indeed let $A$ be an operator acting on the ball centred at the origin with radius $a$. Then, we have that,
\begin{align}
\|\tau_t^{\Lambda_r}(A) - \tau_t^{\Lambda_{r-1}}(A)\| &= 
\left\| \int_0^t ds \frac{d}{dt'}\left(\tau_{t'}^{\Lambda_{r-1}}\left(\tau_{t-t'}^{\Lambda_r}(A)\right)\right)\right\|\\
&= \left\|\int_0^t dt'\tau_{t'}^{\Lambda_{r-1}}\left(\left[H_{\Lambda_r}-H_{\Lambda_{r-1}}, \tau^{\Lambda_r}_{t-t'}(A)\right]\right)\right\|\\
&\leq \int_0^{|t|} dt'\left\| \left[H_{\Lambda_r}-H_{\Lambda_{r-1}}, \tau^{\Lambda_r}_{t'}(A)\right] \right\|\\
&\leq 2\|A\|\|H_{\Lambda_r}-H_{\Lambda_{r-1}}\||\text{supp}(A)| \int_0^{|t|} dt'e^{2s|t'|-\mu(r-a)}\\
&\leq \|A\|P(r)|\text{supp}(A)|\frac{1}{s} e^{-\mu(r-a)+2s|t|}\\
& := C_{LB}P(r)\|A\||\supp(A)| e^{-\mu(r-a)+2s|t|}
\end{align}
Here $C_{LB}$ is a constant and $P(r)$ is a polynomial which is depend on the number of lattice points that are contained in the set $\Lambda_{r} \setminus \Lambda_{r-1}$. 

We can now use the last estimate to show the quasi-locality of the operator $k(s)$. 
\begin{proposition}\label{genlocal}
The generator $K(s)$ of the quasi adiabatic evolution can be written as a sum of quasi-local terms.
\end{proposition}
\begin{proof}
After the previous discussion it suffices to show that the norm of the operators $k_r(s)$ defined in \eqref{decayingterms} decays quickly in $r$. We have that
\begin{align}
\|k_r(s)\| & = \left\|\int_{\mathbb{R}}dt F(\Delta t)\left(\tau_t^{\Lambda_r}(\partial_s h(s))-\tau_t^{\Lambda_{r-1}}(\partial_s h(s))\right) \right\|\\
&\leq 2\int_0^{\infty}dt|F(\Delta t)|\left\|\tau_t^{\Lambda_r}(\partial_s h(s))-\tau_t^{\Lambda_{r-1}}(\partial_s h(s)) \right\|\\
&\leq 2C_NC_SC_{LB}P(r)\int_0^{cr}dt|F(\Delta t)| e^{-\mu(r-C_S/2)+2s|t|} + 4\int_{cr}^{\infty}dt|F(\Delta t)| \|\partial_s h(s)\|\\
&\leq \tilde{C}P(r)\|F\|_{\sup}\frac{1}{s} e^{- r(\mu-2sc)} + 4C_N\int_{cr}^{\infty}dt|F(\Delta t)|.
\end{align}
This last expression clearly vanishes superpolynomially in $r$. Indeed, for the right choice of the constant $c$, say $c \leq \mu/(2s)$ the first term decays exponentially. For a superpolynomially decaying filter function $F$, the last term also decays superpolynomially, although this decay typically only starts for quiet large values of $r$. Notice that this result implies a Lieb-Robinson bound for the interaction $K(s)$  similar to Theorem \ref{theoremLB2}.
\end{proof}
This result remains valid if we start with quasi-local Hamiltonians $H(s)$ as long as the decay function $f$ of $H(s)$ is super-polynomially, $\tilde{f}(R)=\sum_{r>R}f(r)$ decays super-polynomially and the volume of balls of radius $r$ only grows polynomially in $r$. Then, the decay function of the generator $K(s)$ still decreases super-polynomially. Hence, Proposition \ref{genlocal} remains valid for the situation described in Definition \ref{qapath}.
\end{document}